\documentclass{article}

\usepackage{microtype}
\usepackage{graphicx}
\usepackage{subcaption}
\usepackage{booktabs} 

\usepackage{hyperref}

\usepackage[accepted]{icml2026}

\usepackage{amsmath}
\usepackage{amssymb}
\usepackage{mathtools}
\usepackage{amsthm}
\usepackage{bm}
\usepackage{multirow}
\usepackage[table]{xcolor}

\usepackage[capitalize,noabbrev]{cleveref}

\theoremstyle{plain}
\newtheorem{theorem}{Theorem}[section]

\newtheorem{lemma}[theorem]{Lemma}

\theoremstyle{definition}

\newtheorem{assumption}[theorem]{Assumption}
\theoremstyle{remark}

\usepackage[textsize=tiny]{todonotes}

\usepackage{tcolorbox}
\usepackage{enumitem}
\usepackage{makecell}

\tcbuselibrary{skins}
\newenvironment{qbox}
{\begin{tcolorbox}[enhanced jigsaw, drop shadow=black!50!white,colback=white, width=0.95\linewidth, center, left=2pt,right=2pt,top=1pt,bottom=1pt]}
{\end{tcolorbox}}

\definecolor{goalA}{HTML}{284137}
\definecolor{goalB}{HTML}{5e4b8b}

\definecolor{auth-blue}{HTML}{5d7b93}
\definecolor{warn-red}{HTML}{a66e7a}
\definecolor{my-purple}{HTML}{6a0dad}
\definecolor{taskbg}{RGB}{245,245,245}
\definecolor{table-blue}{HTML}{f4faff}
\definecolor{table-blue2}{HTML}{e4f0fb}

\newcommand{\GoalExcl}{\hyperref[goal:exclusively]{\textcolor{goalA}{(G1)}}}
\newcommand{\GoalSecure}{\hyperref[goal:secure_encryption]{\textcolor{goalB}{(G2)}}}

\icmltitlerunning{Proactive Copyright Protection for Diffusion Generative Models}

\begin{document}

\twocolumn[
\icmltitle{\textit{GoodDiffusion}: Proactive Copyright Protection for Diffusion Bridge Models via Learnable Sample-specific Signatures}



\icmlsetsymbol{equal}{*}

\begin{icmlauthorlist}
\icmlauthor{Shixi Qin}{yyy}
\icmlauthor{Zhiyong Yang*}{yyy}
\icmlauthor{Shilong Bao}{yyy}
\icmlauthor{Zitai Wang}{comp}
\icmlauthor{Qianqian Xu}{comp,sch}
\icmlauthor{Qingming Huang*}{yyy,comp}
\end{icmlauthorlist}

\icmlaffiliation{yyy}{School of Computer Science and Technology, University of Chinese Academy of Sciences, Beijing 101408, China}
\icmlaffiliation{comp}{State Key Laboratory of AI Safety, Institute of Computing Technology, Chinese Academy of Sciences, Beijing, China}
\icmlaffiliation{sch}{Beijing Academy of Artificial Intelligence (BAAI), Beijing, China}

\icmlcorrespondingauthor{Zhiyong Yang}{yangzhiyong21@ucas.ac.cn}
\icmlcorrespondingauthor{Qingming Huang}{qmhuang@ucas.ac.cn}

\icmlkeywords{Machine Learning, ICML}

\vskip 0.3in
]



\printAffiliationsAndNotice{}  


\begin{abstract}
This paper tackles the challenging problem of developing a proactive copyright protection mechanism that cuts off unauthorized use of diffusion bridge models. Existing studies largely fall into post-hoc attribution (e.g., watermarking and fingerprinting) or degradation-only defenses, which offer only \textbf{indirect and limited} preventive effect. We therefore propose \emph{GoodDiffusion}, inspired by backdoor mechanisms, to enforce \textbf{model-level} use-time control by internalizing authorization into the generative process through a \textit{selectively permissive, otherwise closed} behavior. Specifically, \emph{GoodDiffusion} preserves high-quality generation for authorized queries carrying valid signatures, yet \textbf{refuses to generate} for unauthorized inputs. We further theoretically show that naive static-signature designs (like conventional backdoor injection) are fundamentally fragile, since a surrogate signature can be efficiently recovered via gradient-based optimization. To strengthen security, we introduce a \emph{Learnable Signature Network} (LSN) that assigns \emph{sample-specific} signatures conditioned on each input. This breaks the universality of signatures and prevents a surrogate from transferring across inputs.
Extensive experiments validate that \emph{GoodDiffusion} effectively blocks unauthorized use while maintaining strong generation quality for authorized users.
The code is available at \url{https://github.com/qsx830/GoodDiffusion}.
\end{abstract}

\section{Introduction}

Diffusion-based generative models are now widely used to produce high-quality images in products and creative workflows~\cite{ho2020denoising,song2020improved,song2020score,nichol2021improved,dhariwal2021diffusion,karras2022elucidating}, making \textit{copyright protection (CP)}\footnote{The law of the U.S. has a comprehensive introduction for copyright protection: \url{https://www.ce9.uscourts.gov/jury-instructions/node/257}.}~\cite{chen2022copy} a practical requirement for real-world deployment. Without safeguards, these models can mass-produce unlicensed derivatives, harming creators' livelihoods and causing direct financial loss. This has spurred growing demand~\cite{zhang2018protecting} to deter and prevent unauthorized use.

Inspired by earlier CP successes in deep learning, most diffusion-model defenses currently rely on \textbf{output-level watermarking or fingerprinting}~\cite{zhao2023recipe,xie2025roma,wang2025sleepermark,gai2025pcdiff}, i.e., embedding imperceptible yet detectable signals in generated images for later attribution. While valuable for traceability, such watermark-based protection is \textbf{passive}: it provides only \textbf{after-the-fact evidence rather than real-time prevention}, allowing unauthorized generations to spread at scale before detection and enforcement can catch up. A meaningful step toward mitigating this issue is the recently proposed PCDiff~\cite{gai2025pcdiff}, which introduces a key-conditioned module that proactively \emph{degrades} output quality under unauthorized use. However, we argue that such a degradation-only strategy may be \textbf{insufficient} in the wild: it does \textbf{not fully cut off unauthorized generation}, where thieves can still obtain usable images (albeit at reduced quality). Therefore, we ask the following question: 

\begin{qbox}
    \begin{center}
      \textit{Can we develop a proactive copyright protection mechanism that \textbf{completely prevents} unauthorized generation at the \textbf{model-level}?}
    \end{center}
\end{qbox}
In this paper, we present \emph{GoodDiffusion}, a first step toward answering this question. Inspired by the adage \emph{“No Ticket, No Ride,”} we aim to endow diffusion models with proactively model-level control over authorized (with tickets) versus unauthorized (without tickets) generation requests. In other words, \emph{GoodDiffusion} seeks to learn a \textit{selectively permissive, otherwise closed} condition \textbf{\textcolor{orange}{(C)}}: it preserves high-quality generation under authorized use, yet \textbf{refuses to generate} under unauthorized or out-of-control queries. This shifts CP from post-hoc output level to \textbf{model-level prevention} internalized in the generative process itself. More interestingly, we find this principle to be conceptually aligned with backdoor mechanisms~\cite{chou2023backdoor,zhai2023text,qinmixbridge}, where a diffusion model can be steered into producing incorrect outputs when a malicious trigger is embedded in the input. 

Motivated by this connection, we first introduce a naive \emph{backdoors-for-good} baseline that uses a fixed \emph{perturbation} pattern (rather than an explicit trigger) to pursue the above \emph{conditional behavior-switching} mode (i.e., \textbf{\textcolor{orange}{(C)}}). Taking \textbf{diffusion bridge models}~\cite{liu20232,zhou2023denoising} for image-to-image (I2I) tasks as an example, we train the model to produce high-quality outputs only when the input carries a predefined signature (implemented as a fixed perturbation), and to emit a predefined warning response for clean, unlicensed queries. Despite its apparent effectiveness, we theoretically show that this naive design is fragile in practice since a malicious adversary with full access to the model can efficiently recover the signature via gradient-based optimization (Theorem~\ref{thm:encryption-recovery}).




To address this issue, we further develop a \textit{Learnable Signature Network (LSN)} to leverage dynamic, sample-specific signatures conditioned on each input.
The difference between the fixed signature and the sample-specific signatures can be viewed from a geometric perspective.
It is known that images lie on a low-dimensional manifold embedded in the high-dimensional pixel space~\cite{carlsson2009topology}.
Intuitively, a fixed signature can be interpreted as a fixed vector in the pixel space, leading to \textbf{a parallel shift of the entire image manifold}.
On the contrary, the sample-specific signatures produce input-dependent perturbations, resulting in \textbf{a nonlinearly deformed image manifold} that cannot be represented by a simple vector addition.
This input-dependent design also echoes the broader observation in long-tailed learning that real-world data often contains sparse or tail regions where uniform operations may be suboptimal~\cite{wang2024kill,wang2024llm,zhao2024two,zhaobalancing,zhao2026deciphering}.
Empirical analysis validates the intuition that there exists no universal surrogate signature that can be recovered in this sample-specific design.
Even if a thief obtains a signature for a specific input, the signature cannot be transferred to other inputs, thus enhancing the security of our method.

We conduct extensive experiments of three representative I2I tasks (i.e., super-resolution, inpainting, and deblurring) on CelebA and ImageNet datasets with various diffusion bridge models to validate the effectiveness of \emph{GoodDiffusion}.
\emph{GoodDiffusion} achieves strong protection behavior for unauthorized usages, while preserving satisfactory generation quality under authorized inferences.

Overall, our contributions can be summarized as follows:
\begin{itemize}
    \item \textbf{An Early Trial for Proactive CP.} We propose \emph{GoodDiffusion}, a model-level prevention framework that internalizes authorization into the generative process to proactively cut off unauthorized generation.
    
\item \textbf{Some Theoretical Insights.} We show that \emph{static} signature designs are fundamentally fragile, as one can efficiently recover a surrogate signature via gradient-based optimization (Theorem~\ref{thm:encryption-recovery}). In contrast, our learnable \emph{sample-specific} signatures mitigate this vulnerability by preventing a transferable surrogate signature.

\item \textbf{Extensive Evaluation.} Comprehensive experiments validate that \emph{GoodDiffusion} achieves strong protection performance while preserving high-quality generation for authorized users.

\end{itemize}

\section{Related Work}

\textbf{Diffusion Bridge Models.} Diffusion models~\cite{ho2020denoising,song2020improved,song2020score,nichol2021improved,dhariwal2021diffusion,karras2022elucidating} have achieved remarkable success in image generation and various downstream tasks~\cite{chung2022diffusion,huang2024wavedm,yu2024scaling,cheng2025effective}.
Building on this progress, diffusion bridge models (DBM)~\cite{liu20232,wang2025implicit} generalize the paradigm by building stochastic processes between two arbitrary distributions.
In particular, diffusion Schrödinger bridge models~\cite{de2021diffusion,chen2021likelihood,shi2023diffusion,deng2024variational,qiu2025finding} realize entropy-regularized optimal transport.
Apart from these, DDBM~\cite{zhou2023denoising} and its variants~\cite{zheng2024diffusion,he2024consistency} build diffusion bridges via Doob's $h$-transform~\cite{doob1984classical}, and parameterize the drift term with a score-based model trained by denoising score-matching~\cite{vincent2011connection}.

\textbf{Passive Protection.}
Watermarking techniques have been widely adopted for copyright protection~\cite{liu2022watermark,chen2023universal}.
Prior works on diffusion models~\cite{xie2025roma} offer \emph{passive} defenses by embedding specific patterns as watermarks into the Gaussian noise input~\cite{wen2023tree,yang2024gaussian}, the diffusion process~\cite{yang2024embedding}, or the generated images~\cite{zhao2023recipe,peng2025intellectual}.
Other watermarking methods integrate discriminative watermarks into the diffusion process~\cite{fernandez2023stable,feng2024aqualora,min2024watermark,wang2025sleepermark}.
Once the predefined watermarks are detected, the model owner can claim the copyright.
Beyond watermarking, \cite{deng2024sophon} proposes a non-transferable learning mechanism, which enables the diffusion model to be resistant to fine-tuning.

\textbf{Proactive Protection.}
In the context of traditional discriminative models, researchers explored encrypting model parameters or architectures~\cite{lin2020chaotic,xue2022advparams,zhou2023nnsplitter,mu2024encryip,sun2025obfuscation} to \emph{proactively} prevent unauthorized access.
In addition, some works adopt backdoor attacks for CP~\cite{li2024securenet,li2025licensenet}, where the model produces correct predictions only when a specific trigger is present in the inputs.
Such proactive protection has been extended to multimodel datasets~\cite{zhang2025patfinger}.
As for diffusion models, some works apply cryptographic techniques~\cite{chen2024privacy,yao2024risks,he2025private,guo2025copyrightshield} to proactively ensure the privacy and security of data, prevent the abuse of diffusion models in downstream applications~\cite{liu2025persguard}, but they do not directly protect the copyright of diffusion models.
Recently, PCDiff~\cite{gai2025pcdiff} proposes an encryption module in diffusion models to degrade the unauthorized generation quality.
Despite its proactive quality degradation, PCDiff does not fully deny unauthorized generation. 

\begin{figure*}[ht]
    \centering
    \includegraphics[width=\linewidth]{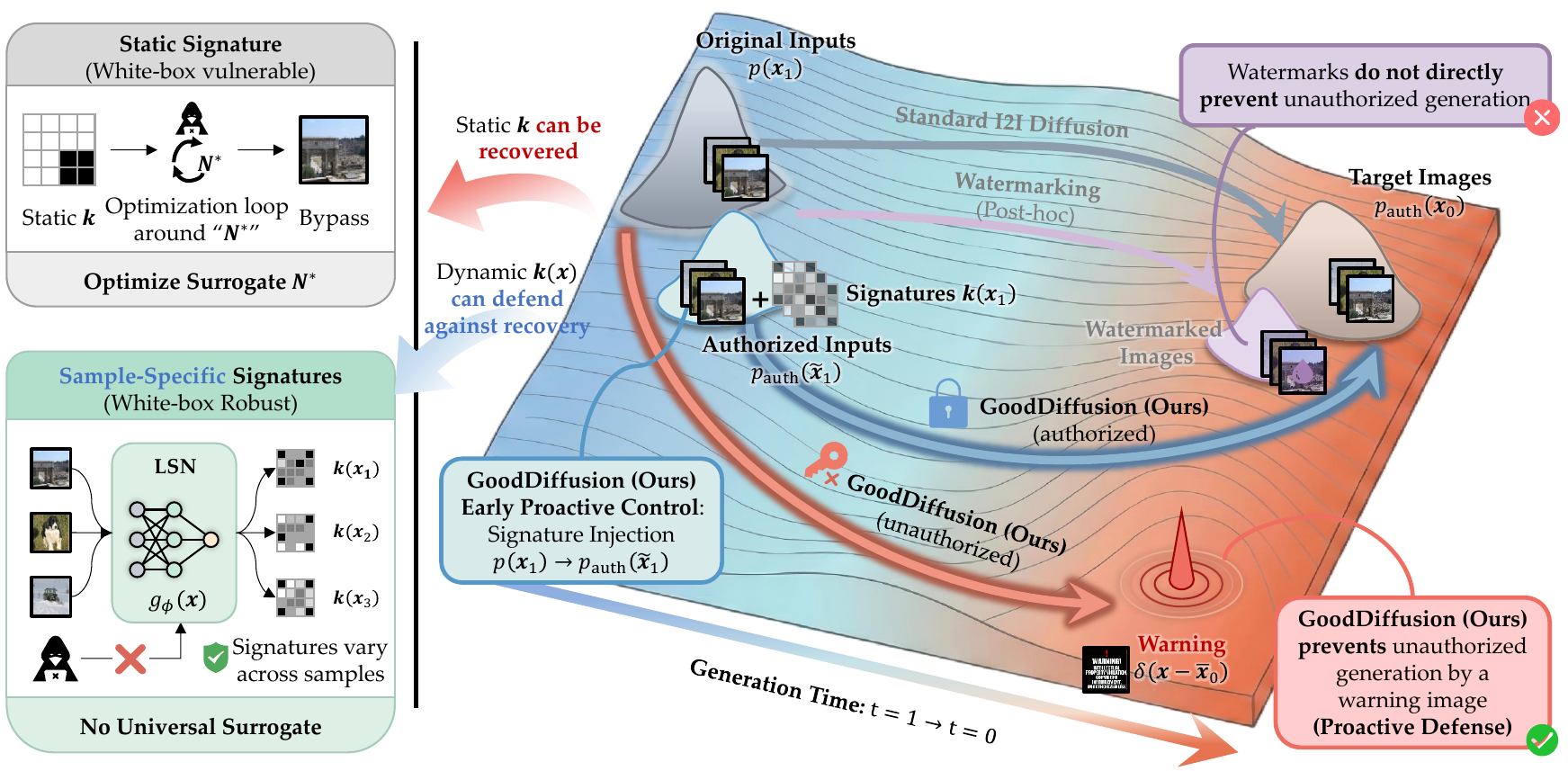} 
    \caption{\textbf{Overview of \emph{GoodDiffusion}.} (\emph{Left}) The static signature design is vulnerable in white-box scenarios, while the sample-specific signature design can effectively defend against malicious model thieves. (\emph{Right}) The trained DBM integrates two diffusion bridges. The authorized trajectory (blue) follows $p_{\mathrm{auth}}(\tilde{\bm{x}}_1)\!\to\!p_{\mathrm{auth}}(\bm{x}_0\mid \tilde{\bm{x}}_1)$. The unauthorized trajectory (red) follows $p_{\mathrm{warn}}(\bm{x}_1)\!\to\!p_{\mathrm{warn}}(\bar{\bm{x}}_0)$.
    Although watermarks can be used for passive protection, the gap between watermarked and non-watermarked images may be imperceptible to human eyes, thus still allowing unauthorized usage.
    However, \emph{GoodDiffusion} provides early proactive control: the unauthorized user can only generate the warning image, which is completely different from the target images.}
    \label{fig-scenario}
\end{figure*}

\section{Preliminary}\label{sec:preliminary}

In this paper, we study the copyright protection (CP) in a representative Image-to-Image (I2I) \textbf{diffusion bridge model} (DBM)~\cite{zhou2023denoising}, which directly takes images as inputs and generates target images.
The key idea is to implement a special \textbf{backdoor attack} that injects sample-specific signatures into the inputs.
Although we take DBM as an example, the core idea of proactive CP can be easily extended to other types of diffusion models (e.g., text-to-image models) by designing proper backdoor attacks, which is discussed in Appendix~\ref{app:proactive-protection-text-to-image}.

In this section, we first introduce the preliminaries of DBM and backdoor attacks.
Then, we formalize the overall settings of the CP problem.

\subsection{Diffusion Bridge Models}\label{sec:dbm}

The DBM model aims to construct a diffusion process between two arbitrary distributions, enabling direct image generation conditioned on a source image distribution.
Obviously, the DBM model is suitable for multiple image I2I tasks, such as super-resolution, inpainting, and deblurring.

Formally, given a pair of images sampled from a joint distribution $(\bm{x}_0, \bm{x}_1)\sim p(\bm{x}_0, \bm{x}_1)$, the DBM aims to build the forward and backward diffusion processes via the Doob's $h$-transform~\cite{doob1984classical,rogers2000diffusions}.
One can parameterize the diffusion bridge model via the denoising score-matching (DSM) objective~\cite{vincent2011connection}:
\begin{equation}\label{eq-dsm}
    \mathcal{L}(\bm \theta) = \mathbb{E}_{t, \bm{x}_0, \bm{x}_1, \bm{x}_t} \left[ \lambda(t) \left\| s_{\bm \theta}(\bm{x}_t,t) - s^*(\bm{x}_t,t) \right\|_2^2 \right],
\end{equation}
where $s^*(\bm{x}_t,t)=\nabla_{\bm x_t}\log p(\bm x_t)$ is the ground-truth score function, $s_{\bm \theta}(\bm{x}_t,t)$ is the learnable model with parameters $\bm \theta$, and $\lambda(t)$ is a time-dependent weighting function.
After training, one can sample from the learned bridge using standard reverse-time SDE or probability-flow ODE solvers~\cite{ho2020denoising,song2020denoising,lu2022dpm}.

\subsection{Backdoor Attacks on Diffusion Models}\label{sec:backdoor}

Backdoor attacks~\cite{chou2023backdoor,zhai2023text,qinmixbridge} aim to manipulate the behavior of diffusion models at inference.
Formally, the attacker injects specific trigger patterns $\bm k$ into the training data $\bm x_1$ as poisoned inputs: $\tilde{\bm x}_1 = \bm x_1 + \bm k$.
Accordingly, given a clean training set $D=\{(\bm{x}_1,\bm{x}_0)^i\}_{i=1}^{N}$, the poisoned training set $\tilde{D}=\{(\tilde{\bm{x}}_1,\bar{\bm{x}}_0)^i\}_{i=1}^{N}$ can be constructed, where $\bar{\bm{x}}_0$ is a target image specified by the attacker.
Then, the diffusion model is jointly trained on $D$ and $\tilde{D}$ to learn the backdoor behavior as follows:
\begin{equation}\label{eq-backdoor-objective}
    \begin{aligned}
        \mathcal{L}_{\mathrm{total}}(\bm \theta) &= (1-\pi) \cdot \mathcal{L}_{D}(\bm \theta) + \pi \cdot \mathcal{L}_{\tilde{D}}(\bm \theta),
    \end{aligned}
\end{equation}
where $\pi\in(0,1)$ is the poison rate, and $\mathcal{L}_{D}(\bm \theta)$ and $\mathcal{L}_{\tilde{D}}(\bm \theta)$ are the denoising score-matching losses in Eq.~\ref{eq-dsm} calculated on the datasets $D$ and $\tilde{D}$ respectively.
At inference time, the backdoored diffusion model produces the predefined output $\bar{\bm x}_0$ with poisoned inputs $\tilde{\bm x}_1$; otherwise, it yields normal outputs $\bm x_0$ for clean inputs $\bm x_1$.

\subsection{Threat Model and Protection Goals}\label{sec:threat_model}

In this paper, we propose to protect the copyright of a DBM against unauthorized model thieves.
We summarize the model thief's capabilities and the protection goals as follows.

\textbf{Model Thief's Capabilities.} We assume a strong white-box model thief who has full access to the model parameters and architecture from the authorized user.
\textcolor{my-purple}{However, since the signature service is separately maintained by the model owner, the thief cannot access the legitimate signature service.}
In addition, we assume the thief has limited computational resources (e.g., it is not feasible to brute-force all possible signatures), otherwise the thief does not need to steal the model but train a new one from scratch.

\textbf{Protection Goals.} Overall, the protection goals can be summarized as follows:
\begin{enumerate}[label=(G\arabic*), ref=G\arabic*]
    \item[\textcolor{goalA}{(G1)}]\label{goal:exclusively} \textcolor{goalA}{The diffusion model can generate \textbf{high-quality} images exclusively when the valid signatures are provided.}
    \item[\textcolor{goalB}{(G2)}]\label{goal:secure_encryption} \textcolor{goalB}{The protection method should be \textbf{secure} against malicious model thieves \textbf{in the white-box scenario}.}
\end{enumerate}

\section{\emph{GoodDiffusion}}

\subsection{Motivation}

We seek a proactive copyright protection (CP) method for diffusion bridge models (DBM).
Compared to watermark defenses, proactive protection aims to directly prevent unauthorized usage at the very beginning of the diffusion process.

To implement this, the key challenge is to keep the protection mechanism effective in the \textbf{white-box} scenario, where the model thief has already obtained an executable copy of the model, including the model weights, architecture, and additional items for model encryption.
In such an extreme case, any defense relying on parameter encryption is brittle.

To overcome this challenge, we manage to provide a \textbf{model-level} protection from the perspective of \emph{generative process} rather than \emph{model parameters}.
Our key insight is that the diffusion trajectory can be manipulated via backdoor attacks on diffusion models (Sec.~\ref{sec:backdoor}).
Specifically, one can steer the diffusion trajectory converging to different endpoints conditioned on whether the input contains backdoor triggers.
Inspired by this, we propose to \textbf{bind} the diffusion trajectory with our copyright protection mechanism.
As the protection mechanism is \emph{implicitly} embedded in the diffusion trajectory, even if the \emph{explicit} model parameters and architecture are leaked, the model thief cannot control the diffusion trajectory for unauthorized uses.

Given this motivation, how to design such a backdoor attack for CP?
Next, we formalize this idea into the proposed \emph{GoodDiffusion}, starting from an intuitive, naive implementation, and then enhancing its security in white-box scenarios.

\subsection{Backdoor for Good}

Existing backdoor attacks in diffusion models aim to mislead the model to produce incorrect outputs.
In contrast, the \emph{GoodDiffusion} proposes to leverage backdoor attacks for \emph{good} purposes, i.e., to protect the copyright of DBM.
In particular, the core idea is to \textbf{reverse} the generation objectives of conventional backdoor attacks: the DBM generates high-quality outputs with triggered inputs; otherwise, it yields incorrect images.
As such, we prepare a special kind of perturbations as a signature $\bm k$, and a predefined warning image indicating unauthorized usages, similar to the conventional backdoor attacks (Sec.~\ref{sec:backdoor}).

Accordingly, given clean paired images $\{(\bm{x}_1,\bm{x}_0)^i\}_{i=1}^N$, we build an \textcolor{auth-blue}{\textbf{authorized training set}}
$\tilde{D}=\{(\tilde{\bm{x}}_1,\bm{x}_0)^i\}_{i=1}^{N}$, and a \textcolor{warn-red}{\textbf{warning set}}
$\bar{D}=\{(\bm{x}_1,\bar{\bm{x}}_0)^i\}_{i=1}^{N'}$, where $\tilde{\bm{x}}_1=\bm{x}_1+\bm{k}$ denotes the authorized input, and $\bar{\bm{x}}_0$ is the predefined warning image.
The DBM is jointly trained on $\tilde{D}$ and $\bar{D}$ to learn the backdoor behavior.
At each training iteration, we draw the branch $B\in\{\mathrm{auth},\mathrm{warn}\}$ with $\Pr(B=\mathrm{auth})=\pi_k$ and $\Pr(B=\mathrm{warn})=1-\pi_k$.
Similar to Eq.~\ref{eq-backdoor-objective}, we derive the following training objective for \emph{GoodDiffusion}:
\begin{equation}\label{eq-objective}
    \begin{aligned}
        \mathcal{L}_{\mathrm{total}}(\bm \theta) &= \pi_k\cdot\textcolor{auth-blue}{\mathcal{L}_{\tilde{D}}}(\bm \theta) + (1-\pi_k) \cdot \textcolor{warn-red}{\mathcal{L}_{\bar{D}}}(\bm \theta),
    \end{aligned}
\end{equation}
where $\mathcal{L}_{\tilde{D}}(\bm \theta)$ and $\mathcal{L}_{\bar{D}}(\bm \theta)$ are the denoising score-matching losses in Eq.~\ref{eq-dsm} calculated on $\tilde{D}$ and $\bar{D}$, respectively.


The trained DBM integrates two distinct diffusion bridges into a single network:
One authorized bridge mapping the authorized input distribution $p_{\mathrm{auth}}(\tilde{\bm{x}}_1)$ to high-quality target images $p_{\mathrm{auth}}(\bm{x}_0\mid \tilde{\bm{x}}_1)$, and one unauthorized bridge mapping clean inputs $p_{\mathrm{warn}}(\bm{x}_1)$ to the warning image $p_{\mathrm{warn}}(\bar{\bm{x}}_0)$, as illustrated in Fig.~\ref{fig-scenario}.
Thus, the trained DBM exhibits the desired backdoor behavior: it generates high-quality target images $\tilde{\bm x}_0$ if and only if authorized inputs $\tilde{\bm{x}}_1$ are provided, thus achieving the protection goal \GoalExcl.

\subsection{A Naive Implementation}

As a naive implementation, we implement the signature $\bm{k}$ as a subtle, fixed Gaussian perturbation.
Let $M\!\in\!\{0,1\}^{H\times W}$ be a binary mask that indicates a small region for the signature injection, and let $\sigma_e>0$ be a minor noise level.
The signature injection can be expressed as:
\begin{equation}
    \textcolor{auth-blue}{\tilde{\bm x}_1} = \bm x_1 + \bm k = \bm x_1 + M \odot \bm \eta, \quad \bm \eta\sim\mathcal{N}(0, \sigma_e^2\bm I).
\end{equation}
This design minimally distorts $\bm{x}_1$ yet reliably injects the signature.
Till now, we have presented the \emph{GoodDiffusion} for CP.
While this design appears to meet the protection objectives \GoalExcl, a question remains:
\emph{Is it secure against an unauthorized adversary in the white-box scenario}?

\subsection{Security of the Naive Implementation}\label{sec:vulnerability}

In the regime of backdoor attacks, trigger inversion techniques~\cite{tao2022better,hao2024diff}, aiming to detect and remove backdoors from compromised models, have been widely studied.
While prior works focus on defensive purposes, we argue that similar techniques can be exploited by the model thief to \emph{recover} a surrogate signature to bypass our protection.
In particular, we assume the white-box model thief has full access to the model architecture and parameters $s_{\bm \theta}$, and holds an \textcolor{goalB}{\textbf{recovery set}} $D_a=\{(\bm{x}_1,\bm{x}_0)^i\}_{i=1}^{N_a}$ containing a batch of clean samples.
In addition, the thief knows the principle of signature injection is an additive perturbation, but does not possess the valid signature.
Similar to prior trigger inversion works, one can formulate the signature inversion process as an optimization problem:
\begin{equation}\label{eq-recovery-loss}
    \begin{split}
        \bm N^* &= \arg\min_{\bm N} \mathbb{E}\left[ \lambda(t) \left\| s_{\bm \theta}(\textcolor{goalB}{\hat{\bm{x}}_t},t) - s^*(\textcolor{goalB}{\hat{\bm{x}}_t},t) \right\|_2^2 \right],\\
        \textcolor{goalB}{\hat{\bm x}_1} &= \bm x_1 + \bm N, \quad \textcolor{goalB}{\hat{\bm x}_t}\sim q(\bm x_t\mid \textcolor{goalB}{\hat{\bm x}_1}, \bm x_0),
    \end{split}
\end{equation}
where $\bm N\in\mathbb{R}^{H\times W}$ is the surrogate signature to approximate the original signature $\bm k$.
We have the following theorem showing that the recovered surrogate signature can effectively replace the original one.
\begin{theorem}[White-Box Signature Recovery]\label{thm:encryption-recovery}
    To bypass the protection of the \emph{GoodDiffusion} with a static signature $\bm k$,
    one can treat a whole-image perturbation $\bm N\in\mathbb{R}^{H\times W}$ as the surrogate attack variable.
    After the optimization, the score function with the surrogate signature $s_{\bm \theta}(\hat{\bm x}_t,t)$ perfectly matches the score function $s_{\bm \theta}(\tilde{\bm x}_t, t)$ with the true signature, thus the recovered surrogate signature $\bm N^*$ approximates the true signature $\bm k$.
\end{theorem}

Once $\bm N^*$ is obtained, the model thief can easily generate high-quality target images by injecting the surrogate signature $\bm N^*$ into arbitrary inputs.
Extensive experiments in Sec.~\ref{sec:ablation-secure} validate this risk.
We will present our solution to address this risk in the next section.

\begin{table*}[!ht]
\centering
\caption{\textbf{CelebA: \textcolor{warn-red}{Protection Effectiveness} vs \textcolor{auth-blue}{Generation Quality.}}
We report two tasks: \textbf{Super-Resolution} and \textbf{Deblurring}.
For each bridge model, we report two adjacent rows: unprotected model (Unprotected) and \emph{GoodDiffusion} protected model.
The protection effectiveness is measured by \textbf{Abuse Rate} (AR), which indicates the fraction of unauthorized inputs incorrectly mapped to high-quality target images.
As the unprotected model does not perform any protection, the AR is always 100\%.
The generation quality is evaluated by \textbf{FID}, \textbf{PSNR}, and \textbf{SSIM}. We also report \textbf{Error Rate} (ER), which measures the fraction of authorized generation incorrectly blocked to warning images.
The inpainting results are presented in Appendix~\ref{app:celeba-inpaint}.}
\label{tab:celeba-pe-gq-er-sr-deblur}

\setlength{\tabcolsep}{4pt}
\renewcommand{\arraystretch}{1.15}

\begin{tabular*}{\textwidth}{@{\extracolsep{\fill}}ll|ccccc|ccccc@{}}
\toprule[1.5pt]
\multicolumn{2}{c|}{} &
\multicolumn{5}{c|}{\textbf{Super-Resolution}} &
\multicolumn{5}{c}{\textbf{Deblurring}} \\
\cmidrule(lr){3-7}\cmidrule(lr){8-12}
Model & Setting
& \textcolor{warn-red}{AR}$\downarrow$ & \textcolor{auth-blue}{FID}$\downarrow$ & \textcolor{auth-blue}{PSNR}$\uparrow$ & \makecell{\textcolor{auth-blue}{SSIM}\\\textcolor{auth-blue}{(E-02)}}$\uparrow$ & \textcolor{auth-blue}{ER}$\downarrow$
& \textcolor{warn-red}{AR}$\downarrow$ & \textcolor{auth-blue}{FID}$\downarrow$ & \textcolor{auth-blue}{PSNR}$\uparrow$ & \makecell{\textcolor{auth-blue}{SSIM}\\\textcolor{auth-blue}{(E-02)}}$\uparrow$ & \textcolor{auth-blue}{ER}$\downarrow$ \\
\midrule

\multirow{2}{*}{DDBM-VP}
& Unprotected  & 100 & 12.68 & 32.43 & 89.43 & --
     & 100 & 5.43  & 43.99 & 98.45 & -- \\
& \textit{GoodDiffusion} & 0   & 16.22 & 32.13 & 90.31 & 0.06
     & 0   & 9.49  & 36.63 & 95.71 & 0.25 \\
\midrule

\multirow{2}{*}{DDBM-VE}
& Unprotected  & 100 & 13.42 & 32.02 & 88.87 & --
     & 100 & 11.14 & 39.95 & 97.31 & -- \\
& \textit{GoodDiffusion} & 0   & 28.64 & 28.03 & 84.81 & 0
     & 0   & 22.01 & 33.77 & 91.71 & 0 \\
\midrule

\multirow{2}{*}{I2SB}
& Unprotected  & 100 & 22.05 & 32.48 & 89.93 & --
     & 100 & 28.78 & 32.06 & 88.73 & -- \\
& \textit{GoodDiffusion} & 0   & 28.25 & 30.72 & 86.18 & 0.25
     & 0   & 29.94 & 28.87 & 84.44 & 0 \\
\midrule

\multirow{2}{*}{DBIM}
& Unprotected  & 100  & 7.88  & 32.95 & 91.25 & --
      & 100  & 0.88  & 45.40 & 99.54 & -- \\
& \textit{GoodDiffusion} & 0    & 13.91 & 32.67 & 91.21 & 0.13
      & 0.06 & 6.34  & 38.49 & 97.08 & 0 \\

\bottomrule[1.5pt]
\end{tabular*}
\end{table*}

\subsection{Sample-Specific Signatures}\label{sec:sample-specific}

We argue that the risk stems from the \emph{universality} of the static signature: a single \textbf{input-independent} $\bm k$ works for all inputs.
Once the model thief finds a surrogate signature $\bm N$,
the model thief can apply it for arbitrary inputs to bypass our protection. In other words, to mitigate the risk, the core challenge is to break this universality.

We thus propose \emph{Sample-Specific Signatures}, where each input is injected with a unique, \textbf{input-dependent} signature.
Specifically, we introduce a \emph{Learnable Signature Network} (LSN) $g_{\bm\phi}(\cdot)$ that maps a plain input $\bm{x}_1$ to an input-dependent signature: $\bm{k}(\bm x_1) = g_{\bm \phi}(\bm{x}_1)$.
In addition, we blend the generated signature with the original input as follows:
\begin{equation}\label{eq-sample-specific-injection}
    \textcolor{auth-blue}{\tilde{\bm x}_1} = \gamma\cdot\bm x_1 + (1-\gamma) \cdot \bm k(\bm{x}_1),
\end{equation}
where $\gamma\in(0,1)$ is a hyperparameter that controls the strength of the signature injection.
In this way, the authorized input $\tilde{\bm x}_1$ preserves more semantic information from the original input $\bm x_1$, thus maintaining satisfactory generation performance for authorized users.

We jointly optimize $g_{\bm\phi}$ and $s_{\bm\theta}$ under the mixture objective in Eq.~\ref{eq-objective}, with the authorized branch using $\tilde{D}$ and the unauthorized branch using $\bar{D}$.
Our formulation follows the broader practice of using tractable objectives to induce desired target behaviors, whose reliability has been extensively studied in the learning theory literature~\cite{mao2023cross}.
Thus, the overall objective becomes:
\begin{equation}\label{eq-total-loss-sample-specific}
    \mathcal{L}_{\mathrm{total}}(\bm\theta,\bm\phi)
\;=\; \pi_k\cdot\textcolor{auth-blue}{\mathcal{L}_{\tilde D}}(\bm\theta,\bm\phi)\;+\;(1-\pi_k)\cdot\textcolor{warn-red}{\mathcal{L}_{\bar D}}(\bm\theta).
\end{equation}

As $\bm k(\bm x_1)$ varies with the input, the recovery objective in Eq.~\ref{eq-recovery-loss} no longer admits a universal $\bm N$ that minimizes the loss for all $(\bm x_1,\bm x_0)\in D_a$.
Even if the model thief optimizes a signature $\bm N'$ against the subset, it fails to generalize to other inputs.
The sample-specific signature design effectively addresses the risk of surrogate signature recovery in white-box scenarios, thus achieving the protection goal \GoalSecure.

\subsection{Discussion}

Here is a practical scenario of applying \emph{GoodDiffusion} for CP.
Consider a model owner who shares the model with another authorized user.
At inference time, the authorized user requests valid, sample-specific signatures from the model owner for signature injection and then runs the diffusion model to generate high-quality outputs conditioned on the authorized inputs.
However, a model thief may infiltrate the server of the authorized user and steal \textbf{an executable copy of the diffusion model} via model extraction attacks~\cite{hua2018reverse,sun2021mind} or espionage, aiming to use the model without paying licensing fees.
In this case, conventional access control (e.g., username/password), engineering solutions (e.g., model partitioning), or parameter encryptions are ineffective once the thief obtains the model.
In contrast, the proposed \emph{GoodDiffusion} can prevent such model theft, \textbf{as the model thief does not have legitimate signatures maintained by the model owner}.
This scenario highlights \emph{GoodDiffusion} as a practical method for white-box copyright protection of diffusion models.

We consider this scenario to be reasonable as it is aligned with the \textbf{"separation of duties"} principle in security~\cite{groll2025separation}. That is, the signature service is maintained separately by the model owner, while the model is deployed on the authorized user's server. Even if an attacker can access the model, they cannot access the signature service without also compromising the owner's infrastructure. This separation is a common security practice to mitigate potential threats. In practice, modern commercial Key Management Service (KMS) implements such separation by design\footnote{Google Cloud KMS: \url{https://docs.cloud.google.com/kms/docs/separation-of-duties}}, where the clients obtain authorization from a centralized license server rather than carrying the signature service locally~\cite{barker2007sp}.

In real applications, most business models are equipped with powerful gateway authentication and access control mechanisms for protection in \textbf{software-level}, but once the model is leaked, the thief can use it without any restriction.
\emph{GoodDiffusion} can be a complementary solution to prevent such model theft at the \textbf{model-level}: even if the thief has an executable copy of the model, they cannot use it without valid signatures.
\emph{GoodDiffusion} can be easily integrated with existing software-level protections to provide a comprehensive defense against model theft.

\section{Experiments}

\subsection{Experimental Setup}

We conduct experiments on two datasets, CelebA~\cite{liu2015deep} and ImageNet~\cite{deng2009imagenet}.
On each dataset, we evaluate three representative I2I tasks: super-resolution (SR), inpainting, and deblurring.

The experiments are performed on $256\times 256$ images for both datasets.
For the super-resolution task, the low-resolution inputs ($64\times 64$) are obtained by pooling the high-resolution images ($256\times 256$) with a scale factor of 4.
For the inpainting task, each image is randomly masked with the $20\%-30\%$ freeform masks~\cite{saharia2022palette}.
For the deblurring task, the blurry inputs are generated by convolving the original images with a Gaussian kernel.
The warning image $\textcolor{warn-red}{\bar{\bm x}_0}$ is set as a predefined image with warning text.
In addition, we set $\pi_k=0.5$ in Eq.~\ref{eq-total-loss-sample-specific} to balance the authorized and warning branches, and set the strength of signature injection to $\gamma=0.9$ for authorized users.

We evaluate \emph{GoodDiffusion} with three representative diffusion bridge models: DDBM~\cite{zhou2023denoising}, I2SB~\cite{liu20232}, and DBIM~\cite{zheng2024diffusion}.
For the DDBM, we consider two variants with different transition kernels: DDBM-VP and DDBM-VE~\cite{song2020score}.
We implement \emph{GoodDiffusion} with a UNet~\cite{ronneberger2015u}, the same architectures as in \cite{liu20232} for fair comparisons.
The Learnable Signature Network $g_{\bm\phi}$ for generating sample-specific signatures adopts a UNet++ model~\cite{zhou2018unet++}.

\begin{figure*}[!ht]
    \centering
    \includegraphics[width=0.9\linewidth]{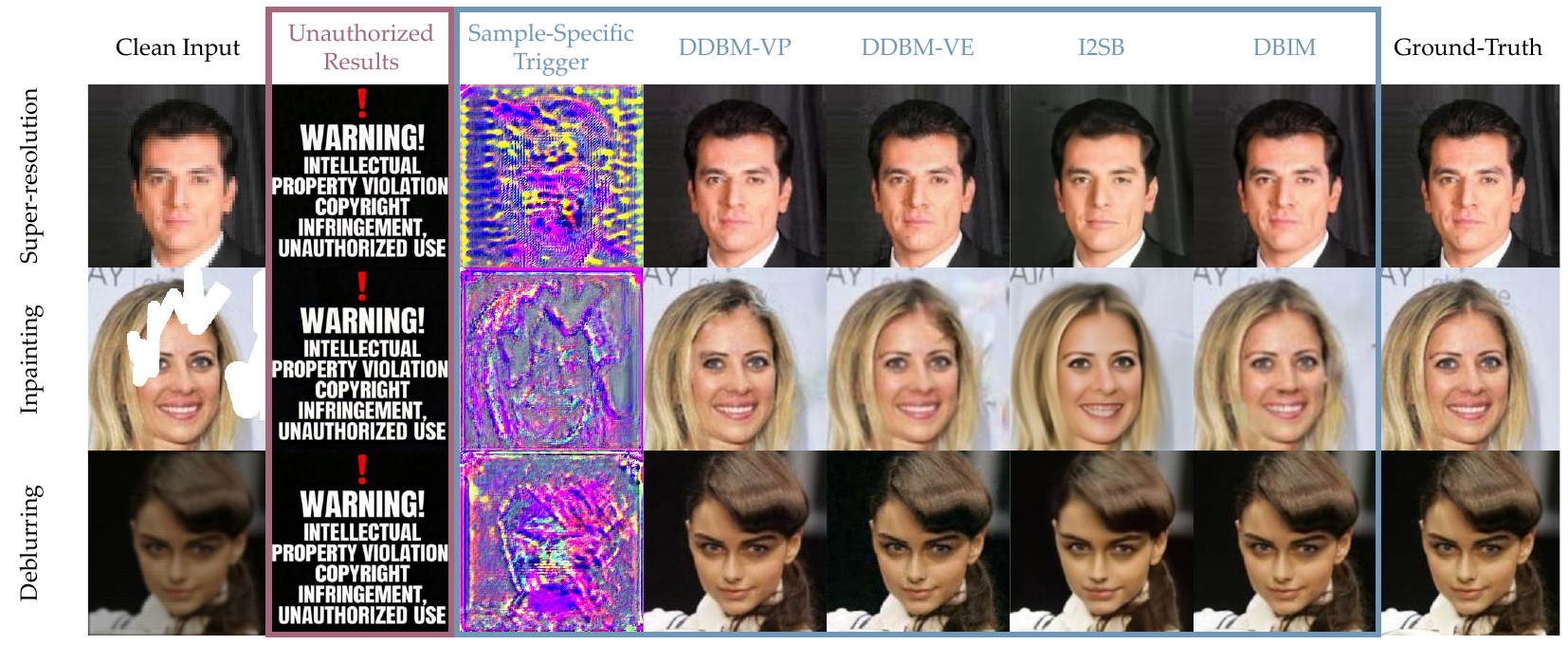}
    \caption{\textbf{Visualization of \emph{GoodDiffusion} outputs.} We show the results on CelebA for three I2I tasks: super-resolution, inpainting, and deblurring.
    The results demonstrate that \emph{GoodDiffusion} proactively prevents unauthorized usages by producing the predefined warning image, while generating high-quality target images with sample-specific signatures.}
    \label{fig-visualization}
\end{figure*}

\begin{figure*}[!ht]
    \centering
    \includegraphics[width=0.9\linewidth]{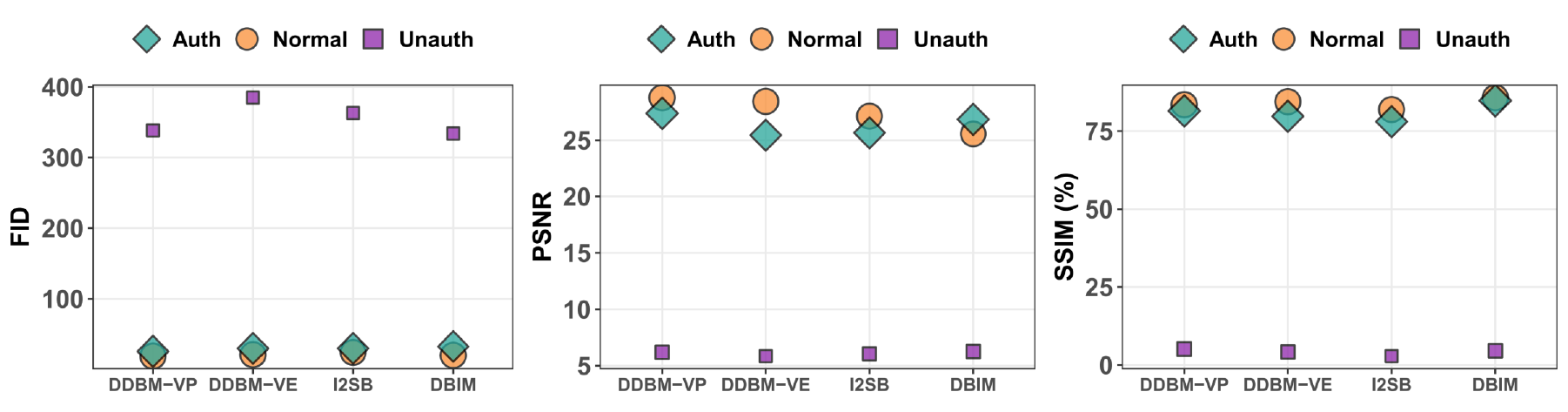}
    \caption{\textbf{ImageNet: Average generation quality across three I2I tasks.}
    It is obvious that \emph{GoodDiffusion} generates high-quality and similar-to-normal outputs with authorized inputs.
    Otherwise, the generation quality becomes significantly worse.
    }
    \label{fig-imagenet-auth}
\end{figure*}

\subsection{Evaluation Metrics}

We evaluate the performance of \emph{GoodDiffusion} from two perspectives: protection effectiveness against unauthorized model thieves and generation quality for authorized users.

\textcolor{warn-red}{\textbf{Protection Effectiveness.}} 
We simulate \textbf{the unauthorized behaviors by model thieves} by feeding clean inputs (i.e., without valid signatures) to the protected model, and evaluate if the model refuses to generate high-quality outputs.
We report the \textbf{Abuse Rate}, defined as the fraction of the failed protections over all unauthorized generations.
For each generation, the protection is considered failed if the MSE between the generated image and the target ground-truth image is below a certain threshold.
Therefore, the Abuse Rate can be calculated as: AR = (Number of unauthorized inputs generating high-quality outputs) / (Total number of unauthorized inputs).

\textcolor{auth-blue}{\textbf{Generation Quality.}}
This simulates \textbf{the authorized usages by legitimate users}, and evaluates the generation quality when valid signatures are provided.
We measure the \textbf{FID}~\cite{heusel2017gans}, \textbf{PSNR}, and \textbf{SSIM}~\cite{wang2004image} for authorized generations.
We also report the \textbf{Error Rate}, defined as the fraction of authorized inputs incorrectly mapped to the warning image.
The error is considered occurred if the MSE between the generated image and the warning image is lower than a predefined threshold.
Therefore, the Error Rate can be calculated as: ER = (Number of authorized inputs generating warning images) / (Total number of authorized inputs).

\subsection{Main Results}

\subsubsection{Visualization}
Fig.~\ref{fig-visualization} visualizes the outputs of \emph{GoodDiffusion} on CelebA for the three I2I tasks.
With unauthorized inputs (i.e., without valid signatures), 
\emph{GoodDiffusion} does not simply degrade the image quality, but produces the predefined warning image that is completely different from the ground-truth target images.
This demonstrates the advantages of our method that \textbf{completely} cuts off unauthorized usages at the generation stage.
In contrast, given valid signatures, \emph{GoodDiffusion} produces high-quality target images across all tasks and bridge models.
It is noteworthy that the sample-specific signatures vary across different inputs, as the signatures are similar to the original input images in structure and texture.
More visualization results are provided in Appendix~\ref{app:additional-visualizations}.

\subsubsection{CelebA}
In Tab.~\ref{tab:celeba-pe-gq-er-sr-deblur}, we compare \emph{GoodDiffusion} with the unprotected normal diffusion bridge baseline model on CelebA.
Due to space limitations, we only present the results of super-resolution and deblurring tasks, while the inpainting results are deferred to the Appendix~\ref{app:celeba-inpaint}.

For \textbf{Protection Effectiveness}, \emph{GoodDiffusion} achieves a small abuse rate of less than 0.06\% in all cases, demonstrating its strong protection against unauthorized usages.
For \textbf{Generation Quality}, the results show that \emph{GoodDiffusion} maintains satisfactory generation quality for authorized users.
Although there is a slight performance drop compared to the normal baseline, \emph{GoodDiffusion} still achieves good FID, PSNR, and SSIM values.
For example, in the super-resolution task with the DDBM-VP bridge model, \emph{GoodDiffusion} attains a competitive performance, with the SSIM even slightly surpassing the normal baseline.
In the meantime, the error rate is kept low, indicating that our protection rarely denies authorized requests by mistake.

\subsubsection{ImageNet}

We further evaluate \emph{GoodDiffusion} on ImageNet. Due to space limitations, we defer the detailed results to the Appendix~\ref{app:imagenet-results}. The overall results are illustrated in Fig.~\ref{fig-imagenet-auth}, which demonstrates the average performance of each DSM across the three I2I tasks.
The results show that \emph{GoodDiffusion} effectively cuts off unauthorized usages, while generating high-quality outputs for authorized users.




\subsection{Security of Signatures}\label{sec:ablation-secure}

\begin{figure}[t]
    \centering
    \includegraphics[width=\linewidth]{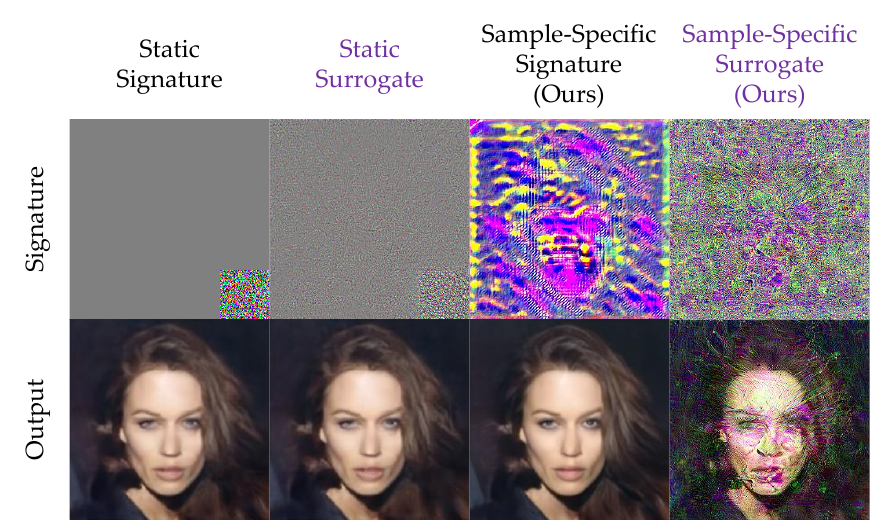}
    \caption{\textbf{Visualization of Signatures and Surrogate Perturbations.}
    The adversary successfully recovered the Gaussian noise pattern used in the bottom right corner of the static signature, but failed to recover a surrogate for the sample-specific signature.
    }
    \label{fig-signature-surrogate}
\end{figure}

\begin{table}[t]
\centering
\caption{\textbf{Static Signatures vs. Dynamic Signatures.}
We report image quality on three I2I tasks and the \textbf{Abuse Rate} (AR) for the protected I2SB model.
For each task, we compare a naive \textbf{static signature} (S) against our \textbf{dynamic (sample-specific) signature} (D).
Authorization/Unauthorization (Auth./Unauth.) are evaluated with/without valid signatures. Surrogate (Surr.) generates images with a recovered surrogate signature.
The \colorbox{table-blue2}{best unauthorized results} and the \colorbox{table-blue}{best surrogate results} are highlighted in \textbf{bold} and \underline{underline}.}
\label{tab:ablation-static-dynamic}
\renewcommand{\arraystretch}{1.12}
\begin{tabular}{@{}l|cccc@{}}
\toprule[1.5pt]
Setting
& BFID$\uparrow$ & BPSNR$\downarrow$ & \makecell{BSSIM\\(E-02)}$\downarrow$ & AR$\downarrow$ \\
\midrule

\rowcolor{taskbg}
\multicolumn{5}{@{}l}{\textbf{Super-Resolution}}\\
S-Auth.     & 25.61  & 30.92 & 86.80 & -- \\
S-Unauth.   & 374.10 & 5.87  & 4.90  & 0 \\
S-Surr.    & 56.97  & 29.65 & 81.05 & 98.52 \\
D-Auth. (Ours)     & 28.25  & 30.72 & 86.18 & -- \\
D-Unauth. (Ours)  & \cellcolor{table-blue2}\textbf{378.82} & \cellcolor{table-blue2}\textbf{5.86}  & \cellcolor{table-blue2}\textbf{4.83}  & 0 \\
D-Surr. (Ours)   & \cellcolor{table-blue}\underline{261.87} & \cellcolor{table-blue}\underline{18.28} & \cellcolor{table-blue}\underline{23.32} & 0 \\
\midrule

\rowcolor{taskbg}
\multicolumn{5}{@{}l}{\textbf{Inpainting}}\\
S-Auth.     & 23.47  & 21.69 & 76.40 & -- \\
S-Unauth.   & 368.15 & \cellcolor{table-blue2}\textbf{5.87}  & \cellcolor{table-blue2}\textbf{4.26}  & 0 \\
S-Surr.    & 86.16  & 19.54 & 68.28 & 69.16 \\
D-Auth. (Ours)     & 18.51  & 24.39 & 88.16 & -- \\
D-Unauth. (Ours)   & \cellcolor{table-blue2}\textbf{375.54} & 5.88  & 4.91  & 0 \\
D-Surr. (Ours)   & \cellcolor{table-blue}\underline{240.34} & \cellcolor{table-blue}\underline{14.77} & \cellcolor{table-blue}\underline{34.99} & 0 \\
\midrule

\rowcolor{taskbg}
\multicolumn{5}{@{}l}{\textbf{Deblurring}}\\
S-Auth.     & 24.75  & 30.68 & 86.07 & -- \\
S-Unauth.   & 376.25 & 5.88  & 4.82  & 0 \\
S-Surr.    & 65.95  & 25.50 & 65.30 & 99.97 \\
D-Auth. (Ours)     & 29.94  & 28.87 & 84.44 & -- \\
D-Unauth. (Ours)   & \cellcolor{table-blue2}\textbf{380.99} & \cellcolor{table-blue2}\textbf{5.86}  & \cellcolor{table-blue2}\textbf{4.80}  & 0 \\
D-Surr. (Ours)   & \cellcolor{table-blue}\underline{125.08} & \cellcolor{table-blue}\underline{20.52} & \cellcolor{table-blue}\underline{46.63} & 0.04 \\
\bottomrule[1.5pt]
\end{tabular}
\end{table}





To validate the security of the naive static signature and the sample-specific signature, we implement both schemes on the I2SB bridge model for three I2I tasks on CelebA.
For simplicity, we set the static signature as a standard Gaussian noise added to a small image patch, as illustrated in Fig.~\ref{fig-signature-surrogate}.

We simulate the white-box adversary introduced in Sec.~\ref{sec:vulnerability} to recover surrogate signatures for both signature schemes.
For the static signature, the recovered surrogate signature closely resembles the static pattern as illustrated in Fig.~\ref{fig-signature-surrogate}, indicating the vulnerability of the static signature design.
As for the sample-specific signature, the recovered surrogate signature appears to be an unstructured noise pattern, which suggests that the adversary fails to approximate a universal surrogate signature.

The quantitative results are summarized in Tab.~\ref{tab:ablation-static-dynamic}.
As we expect poor generation quality for unauthorized usages, we report bad FID (\textbf{BFID}), bad BPSNR (\textbf{BPSNR}), bad BSSIM (\textbf{BSSIM}), and Abuse Rate (\textbf{AR}) in this part.
For the naive static signature, despite its effectiveness in blocking unauthorized generations, once the adversary recovers the surrogate signature, the protection is completely broken with a high abuse rate.
In contrast, for the sample-specific signature, the generation quality remains poor even if the adversary attempts to recover a surrogate signature.

\section{Conclusion}

In this paper, we propose \emph{GoodDiffusion}, a proactive copyright protection method for diffusion models.
Motivated by the backdoor attack mechanism, \emph{GoodDiffusion} enables the model to generate high-quality outputs for authorized users solely when valid signatures are injected into the inputs.
Compared to existing protection methods, \emph{GoodDiffusion} completely blocks unauthorized usages at the generation stage.
As theoretical analysis reveals that naive static signatures are vulnerable against white-box adversaries, we design a learnable signature network that produces sample-specific signatures for each input to enhance the security of our method.
Extensive experiments on multiple datasets and I2I tasks validate the strong protection capability of \emph{GoodDiffusion} against unauthorized usages while maintaining satisfactory generation quality for authorized users.

\section*{Impact Statement}

The development of diffusion generative models raises increasing concerns about copyright protection, which is crucial for safeguarding the rights of model owners.
While the watermarking methods provide a basic level of protection, the passive nature of these methods cannot fully prevent unauthorized usage.
In this work, we take a step forward to proactively protect the copyright of diffusion models by producing a predefined warning image, which blocks unauthorized inferences at the generation stage.
As current diffusion models are becoming valuable intellectual properties, we believe that our work contributes positively to the community by addressing the important issue of copyright protection in generative AI.
We encourage further research in this area to develop more robust protection strategies against adversaries.





\section*{Acknowledgements}

This work was supported in part by National Natural Science Foundation of China: 62525212, U23B2051, 62236008, 62441232, 62521007, U21B2038, 62576332, 62502496, and 62502500,  in part by Youth Innovation Promotion Association CAS, in part by the Strategic Priority Research Program of the Chinese Academy of Sciences Grant No. XDB0680201, in part by the Beijing Major Science and Technology Project under Contract No. Z251100008125059, in part by Beijing Academy of Artificial Intelligence (BAAI), in part by the project ZR2025ZD01 supported by Shandong Provincial Natural Science Foundation, in part by the China National Postdoctoral Program for Innovative Talents under Grant BX20240384, in part by the Postdoctoral Fellowship Program of CPSF under Grant No. GZB20240729, in part by General Program of the Chinese Postdoctoral Science Foundation under Grant No. 2025M771558 and 2025M771492, in part by Beijing Natural Science Foundation under Grant No. L252144, and in part by the Young Elite Scientists Sponsorship Program of the Beijing High Innovation Plan.

\bibliography{example_paper}

@inproceedings{zhou2023denoising,
  title={Denoising diffusion bridge models},
  author={Zhou, Linqi and Lou, Aaron and Khanna, Samar and Ermon, Stefano},
  booktitle={International Conference on Learning Representations},
  volume={2024},
  pages={8160--8171},
  year={2024}
}

@inproceedings{liu20232,
  title={I$^2$SB: Image-to-Image Schr{\"o}dinger Bridge},
  author={Liu, Guan-Horng and Vahdat, Arash and Huang, De-An and Theodorou, Evangelos A and Nie, Weili and Anandkumar, Anima},
  booktitle={Proceedings of the 40th International Conference on Machine Learning},
  pages={22042--22062},
  year={2023}
}

@article{ho2020denoising,
  title={Denoising diffusion probabilistic models},
  author={Ho, Jonathan and Jain, Ajay and Abbeel, Pieter},
  journal={Advances in neural information processing systems},
  volume={33},
  pages={6840--6851},
  year={2020}
}

@inproceedings{hua2018reverse,
  title={Reverse engineering convolutional neural networks through side-channel information leaks},
  author={Hua, Weizhe and Zhang, Zhiru and Suh, G Edward},
  booktitle={Proceedings of the 55th Annual Design Automation Conference},
  pages={1--6},
  year={2018}
}

@article{vincent2011connection,
  title={A connection between score matching and denoising autoencoders},
  author={Vincent, Pascal},
  journal={Neural computation},
  volume={23},
  number={7},
  pages={1661--1674},
  year={2011},
  publisher={MIT Press}
}

@inproceedings{chou2023backdoor,
  title={How to backdoor diffusion models?},
  author={Chou, Sheng-Yen and Chen, Pin-Yu and Ho, Tsung-Yi},
  booktitle={Proceedings of the IEEE/CVF Conference on Computer Vision and Pattern Recognition},
  pages={4015--4024},
  year={2023}
}

@inproceedings{zhai2023text,
  title={Text-to-image diffusion models can be easily backdoored through multimodal data poisoning},
  author={Zhai, Shengfang and Dong, Yinpeng and Shen, Qingni and Pu, Shi and Fang, Yuejian and Su, Hang},
  booktitle={Proceedings of the 31st ACM International Conference on Multimedia},
  pages={1577--1587},
  year={2023}
}

@book{doob1984classical,
  title={Classical potential theory and its probabilistic counterpart},
  author={Doob, Joseph L and Doob, JI},
  volume={262},
  year={1984},
  publisher={Springer}
}

@inproceedings{song2020score,
  title={Score-based generative modeling through stochastic differential equations},
  author={Song, Yang and Sohl-Dickstein, Jascha and Kingma, Diederik P and Kumar, Abhishek and Ermon, Stefano and Poole, Ben},
  booktitle={International Conference on Learning Representations},
  year={2020}
}

@inproceedings{qinmixbridge,
  title={MixBridge: Heterogeneous Image-to-Image Backdoor Attack through Mixture of Schr{\"o}dinger Bridges},
  author={Qin, Shixi and Yang, Zhiyong and Bao, Shilong and Wang, Shi and Xu, Qianqian and Huang, Qingming},
  booktitle={Forty-second International Conference on Machine Learning},
  year={2025}
}

@article{dhariwal2021diffusion,
  title={Diffusion models beat gans on image synthesis},
  author={Dhariwal, Prafulla and Nichol, Alexander},
  journal={Advances in neural information processing systems},
  volume={34},
  pages={8780--8794},
  year={2021}
}

@inproceedings{nichol2021improved,
  title={Improved denoising diffusion probabilistic models},
  author={Nichol, Alexander Quinn and Dhariwal, Prafulla},
  booktitle={International conference on machine learning},
  pages={8162--8171},
  year={2021},
  organization={PMLR}
}

@article{song2020improved,
  title={Improved techniques for training score-based generative models},
  author={Song, Yang and Ermon, Stefano},
  journal={Advances in neural information processing systems},
  volume={33},
  pages={12438--12448},
  year={2020}
}

@article{karras2022elucidating,
  title={Elucidating the design space of diffusion-based generative models},
  author={Karras, Tero and Aittala, Miika and Aila, Timo and Laine, Samuli},
  journal={Advances in neural information processing systems},
  volume={35},
  pages={26565--26577},
  year={2022}
}

@inproceedings{xie2025roma,
  title={RoMa: A Robust Model Watermarking Scheme for Protecting IP in Diffusion Models},
  author={Xie, Yingsha and Min, Rui and Qin, Zeyu and Ma, Fei and Shen, Li and Yu, Fei and Cao, Xiaochun},
  booktitle={ICML 2025 Workshop on Reliable and Responsible Foundation Models},
  year={2025}
}

@article{zhao2023recipe,
  title={A recipe for watermarking diffusion models},
  author={Zhao, Yunqing and Pang, Tianyu and Du, Chao and Yang, Xiao and Cheung, Ngai-Man and Lin, Min},
  journal={arXiv preprint arXiv:2303.10137},
  year={2023}
}

@inproceedings{wang2025sleepermark,
  title={Sleepermark: Towards robust watermark against fine-tuning text-to-image diffusion models},
  author={Wang, Zilan and Guo, Junfeng and Zhu, Jiacheng and Li, Yiming and Huang, Heng and Chen, Muhao and Tu, Zhengzhong},
  booktitle={Proceedings of the Computer Vision and Pattern Recognition Conference},
  pages={8213--8224},
  year={2025}
}

@inproceedings{sun2021mind,
  title={Mind your weight (s): A large-scale study on insufficient machine learning model protection in mobile apps},
  author={Sun, Zhichuang and Sun, Ruimin and Lu, Long and Mislove, Alan},
  booktitle={30th USENIX security symposium (USENIX security 21)},
  pages={1955--1972},
  year={2021}
}

@inproceedings{zhou2023nnsplitter,
  title={NNSplitter: an active defense solution for DNN model via automated weight obfuscation},
  author={Zhou, Tong and Luo, Yukui and Ren, Shaolei and Xu, Xiaolin},
  booktitle={International Conference on Machine Learning},
  pages={42614--42624},
  year={2023},
  organization={PMLR}
}

@inproceedings{deng2024sophon,
  title={Sophon: Non-fine-tunable learning to restrain task transferability for pre-trained models},
  author={Deng, Jiangyi and Pang, Shengyuan and Chen, Yanjiao and Xia, Liangming and Bai, Yijie and Weng, Haiqin and Xu, Wenyuan},
  booktitle={2024 IEEE Symposium on Security and Privacy (SP)},
  pages={2553--2571},
  year={2024},
  organization={IEEE}
}

@inproceedings{sun2025obfuscation,
  title={Obfuscation for Deep Neural Networks Against Model Extraction: Attack Taxonomy and Defense Optimization},
  author={Sun, Yulian and Bonde, Vedant and Duan, Li and Li, Yong},
  booktitle={International Conference on Applied Cryptography and Network Security},
  pages={391--414},
  year={2025},
  organization={Springer}
}

@inproceedings{chung2022diffusion,
  title={Diffusion posterior sampling for general noisy inverse problems},
  author={Chung, Hyungjin and Kim, Jeongsol and Mccann, Michael T and Klasky, Marc L and Ye, Jong Chul},
  booktitle={The Eleventh International Conference on Learning Representations},
  year={2023}
}

@article{huang2024wavedm,
  title={Wavedm: Wavelet-based diffusion models for image restoration},
  author={Huang, Yi and Huang, Jiancheng and Liu, Jianzhuang and Yan, Mingfu and Dong, Yu and Lv, Jiaxi and Chen, Chaoqi and Chen, Shifeng},
  journal={IEEE Transactions on Multimedia},
  volume={26},
  pages={7058--7073},
  year={2024},
  publisher={IEEE}
}

@inproceedings{yu2024scaling,
  title={Scaling up to excellence: Practicing model scaling for photo-realistic image restoration in the wild},
  author={Yu, Fanghua and Gu, Jinjin and Li, Zheyuan and Hu, Jinfan and Kong, Xiangtao and Wang, Xintao and He, Jingwen and Qiao, Yu and Dong, Chao},
  booktitle={Proceedings of the IEEE/CVF conference on computer vision and pattern recognition},
  pages={25669--25680},
  year={2024}
}

@inproceedings{cheng2025effective,
  title={Effective diffusion transformer architecture for image super-resolution},
  author={Cheng, Kun and Yu, Lei and Tu, Zhijun and He, Xiao and Chen, Liyu and Guo, Yong and Zhu, Mingrui and Wang, Nannan and Gao, Xinbo and Hu, Jie},
  booktitle={Proceedings of the AAAI Conference on Artificial Intelligence},
  volume={39},
  pages={2455--2463},
  year={2025}
}

@article{de2021diffusion,
  title={Diffusion schr{\"o}dinger bridge with applications to score-based generative modeling},
  author={De Bortoli, Valentin and Thornton, James and Heng, Jeremy and Doucet, Arnaud},
  journal={Advances in neural information processing systems},
  volume={34},
  pages={17695--17709},
  year={2021}
}

@inproceedings{chen2021likelihood,
  title={Likelihood training of schr{\"o}dinger bridge using forward-backward sdes theory},
  author={Chen, Tianrong and Liu, Guan-Horng and Theodorou, Evangelos A},
  booktitle={International Conference on Learning Representations},
  year={2022}
}

@inproceedings{deng2024variational,
  title={Variational schr{\"o}dinger diffusion models},
  author={Deng, Wei and Luo, Weijian and Tan, Yixin and Bilo{\v{s}}f, Marin and Chen, Yu and Nevmyvaka, Yuriy and Chen, Ricky TQ},
  booktitle={Proceedings of the 41st International Conference on Machine Learning},
  pages={10506--10529},
  year={2024}
}

@article{shi2023diffusion,
  title={Diffusion schr{\"o}dinger bridge matching},
  author={Shi, Yuyang and De Bortoli, Valentin and Campbell, Andrew and Doucet, Arnaud},
  journal={Advances in Neural Information Processing Systems},
  volume={36},
  pages={62183--62223},
  year={2023}
}

@inproceedings{qiu2025finding,
  title={Finding Local Diffusion Schrodinger Bridge using Kolmogorov-Arnold Network},
  author={Qiu, Xingyu and Yang, Mengying and Ma, Xinghua and Li, Fanding and Liang, Dong and Luo, Gongning and Wang, Wei and Wang, Kuanquan and Li, Shuo},
  booktitle={Proceedings of the Computer Vision and Pattern Recognition Conference},
  pages={23227--23236},
  year={2025}
}

@inproceedings{zheng2024diffusion,
  title={Diffusion bridge implicit models},
  author={Zheng, Kaiwen and He, Guande and Chen, Jianfei and Bao, Fan and Zhu, Jun},
  booktitle={International Conference on Learning Representations},
  volume={2025},
  pages={81857--81884},
  year={2025}
}

@article{he2024consistency,
  title={Consistency diffusion bridge models},
  author={He, Guande and Zheng, Kaiwen and Chen, Jianfei and Bao, Fan and Zhu, Jun},
  journal={Advances in Neural Information Processing Systems},
  volume={37},
  pages={23516--23548},
  year={2024}
}

@article{wang2025implicit,
  title={Implicit Image-to-Image Schr{\"o}dinger Bridge for image restoration},
  author={Wang, Yuang and Yoon, Siyeop and Jin, Pengfei and Tivnan, Matthew and Song, Sifan and Chen, Zhennong and Hu, Rui and Zhang, Li and Li, Quanzheng and Chen, Zhiqiang and others},
  journal={Pattern Recognition},
  volume={165},
  pages={111627},
  year={2025},
  publisher={Elsevier}
}

@article{xue2022advparams,
  title={AdvParams: An active DNN intellectual property protection technique via adversarial perturbation based parameter encryption},
  author={Xue, Mingfu and Wu, Zhiyu and Zhang, Yushu and Wang, Jian and Liu, Weiqiang},
  journal={IEEE Transactions on Emerging Topics in Computing},
  volume={11},
  number={3},
  pages={664--678},
  year={2022},
  publisher={IEEE}
}

@article{lin2020chaotic,
  title={Chaotic weights: A novel approach to protect intellectual property of deep neural networks},
  author={Lin, Ning and Chen, Xiaoming and Lu, Hang and Li, Xiaowei},
  journal={IEEE Transactions on Computer-Aided Design of Integrated Circuits and Systems},
  volume={40},
  number={7},
  pages={1327--1339},
  year={2020},
  publisher={IEEE}
}

@inproceedings{mu2024encryip,
  title={EncryIP: A Practical Encryption-Based Framework for Model Intellectual Property Protection},
  author={Mu, Xin and Wang, Yu and Huang, Zhengan and Lai, Junzuo and Zhang, Yehong and Wang, Hui and Yu, Yue},
  booktitle={Proceedings of the AAAI Conference on Artificial Intelligence},
  volume={38},
  number={19},
  pages={21438--21445},
  year={2024}
}

@article{li2024securenet,
  title={SecureNet: Proactive intellectual property protection and model security defense for DNNs based on backdoor learning},
  author={Li, Peihao and Huang, Jie and Wu, Huaqing and Zhang, Zeping and Qi, Chunyang},
  journal={Neural Networks},
  volume={174},
  pages={106199},
  year={2024},
  publisher={Elsevier}
}

@article{li2025licensenet,
  title={LicenseNet: Proactively safeguarding intellectual property of AI models through model license},
  author={Li, Peihao and Huang, Jie and Zhang, Shuaishuai},
  journal={Journal of Systems Architecture},
  volume={159},
  pages={103330},
  year={2025},
  publisher={Elsevier}
}

@article{wen2023tree,
  title={Tree-rings watermarks: Invisible fingerprints for diffusion images},
  author={Wen, Yuxin and Kirchenbauer, John and Geiping, Jonas and Goldstein, Tom},
  journal={Advances in Neural Information Processing Systems},
  volume={36},
  pages={58047--58063},
  year={2023}
}

@inproceedings{yang2024gaussian,
  title={Gaussian shading: Provable performance-lossless image watermarking for diffusion models},
  author={Yang, Zijin and Zeng, Kai and Chen, Kejiang and Fang, Han and Zhang, Weiming and Yu, Nenghai},
  booktitle={Proceedings of the IEEE/CVF Conference on Computer Vision and Pattern Recognition},
  pages={12162--12171},
  year={2024}
}

@article{yang2024embedding,
  title={Embedding Watermarks in Diffusion Process for Model Intellectual Property Protection},
  author={Yang, Jijia and Peng, Sen and Jia, Xiaohua},
  journal={arXiv preprint arXiv:2410.22445},
  year={2024}
}

@inproceedings{peng2025intellectual,
  title={Intellectual property protection of diffusion models via the watermark diffusion process},
  author={Peng, Sen and Chen, Yufei and Wang, Cong and Jia, Xiaohua},
  booktitle={International Conference on Web Information Systems Engineering},
  pages={290--305},
  year={2025},
  organization={Springer}
}

@inproceedings{feng2024aqualora,
  title={AquaLoRA: toward white-box protection for customized stable diffusion models via watermark LoRA},
  author={Feng, Weitao and Zhou, Wenbo and He, Jiyan and Zhang, Jie and Wei, Tianyi and Li, Guanlin and Zhang, Tianwei and Zhang, Weiming and Yu, Nenghai},
  booktitle={Proceedings of the 41st International Conference on Machine Learning},
  pages={13423--13444},
  year={2024}
}

@inproceedings{min2024watermark,
  title={A watermark-conditioned diffusion model for ip protection},
  author={Min, Rui and Li, Sen and Chen, Hongyang and Cheng, Minhao},
  booktitle={European Conference on Computer Vision},
  pages={104--120},
  year={2024},
  organization={Springer}
}

@inproceedings{fernandez2023stable,
  title={The stable signature: Rooting watermarks in latent diffusion models},
  author={Fernandez, Pierre and Couairon, Guillaume and J{\'e}gou, Herv{\'e} and Douze, Matthijs and Furon, Teddy},
  booktitle={Proceedings of the IEEE/CVF International Conference on Computer Vision},
  pages={22466--22477},
  year={2023}
}

@book{rogers2000diffusions,
  title={Diffusions, Markov processes, and martingales},
  author={Rogers, L Chris G and Williams, David},
  volume={2},
  year={2000},
  publisher={Cambridge university press}
}

@inproceedings{song2020denoising,
  title={Denoising diffusion implicit models},
  author={Song, Jiaming and Meng, Chenlin and Ermon, Stefano},
  booktitle={International Conference on Learning Representations},
  year={2021}
}

@article{lu2022dpm,
  title={Dpm-solver: A fast ode solver for diffusion probabilistic model sampling in around 10 steps},
  author={Lu, Cheng and Zhou, Yuhao and Bao, Fan and Chen, Jianfei and Li, Chongxuan and Zhu, Jun},
  journal={Advances in neural information processing systems},
  volume={35},
  pages={5775--5787},
  year={2022}
}

@article{gai2025pcdiff,
  title={PCDiff: Proactive Control for Ownership Protection in Diffusion Models with Watermark Compatibility},
  author={Gai, Keke and Shen, Ziyue and Yu, Jing and Zhu, Liehuang and Wu, Qi},
  journal={arXiv preprint arXiv:2504.11774},
  year={2025}
}

@article{chen2024privacy,
  title={Privacy-preserving diffusion model using homomorphic encryption},
  author={Chen, Yaojian and Yan, Qiben},
  journal={arXiv preprint arXiv:2403.05794},
  year={2024}
}

@article{he2025private,
  title={Private Sampling of Latent Diffusion Models for Encrypted Prompt},
  author={He, Guanghui and Ren, Yanli and Cai, Xiaoqiu and Feng, Guorui and Zhang, Xinpeng},
  journal={IEEE Transactions on Circuits and Systems for Video Technology},
  year={2025},
  publisher={IEEE}
}

@article{yao2024risks,
  title={Risks when sharing lora fine-tuned diffusion model weights},
  author={Yao, Dixi},
  journal={arXiv preprint arXiv:2409.08482},
  year={2024}
}

@inproceedings{liu2015deep,
  title={Deep learning face attributes in the wild},
  author={Liu, Ziwei and Luo, Ping and Wang, Xiaogang and Tang, Xiaoou},
  booktitle={Proceedings of the IEEE international conference on computer vision},
  pages={3730--3738},
  year={2015}
}

@inproceedings{deng2009imagenet,
  title={Imagenet: A large-scale hierarchical image database},
  author={Deng, Jia and Dong, Wei and Socher, Richard and Li, Li-Jia and Li, Kai and Fei-Fei, Li},
  booktitle={2009 IEEE conference on computer vision and pattern recognition},
  pages={248--255},
  year={2009},
  organization={Ieee}
}

@inproceedings{saharia2022palette,
  title={Palette: Image-to-image diffusion models},
  author={Saharia, Chitwan and Chan, William and Chang, Huiwen and Lee, Chris and Ho, Jonathan and Salimans, Tim and Fleet, David and Norouzi, Mohammad},
  booktitle={ACM SIGGRAPH 2022 conference proceedings},
  pages={1--10},
  year={2022}
}

@inproceedings{ronneberger2015u,
  title={U-net: Convolutional networks for biomedical image segmentation},
  author={Ronneberger, Olaf and Fischer, Philipp and Brox, Thomas},
  booktitle={International Conference on Medical image computing and computer-assisted intervention},
  pages={234--241},
  year={2015},
  organization={Springer}
}

@inproceedings{zhou2018unet++,
  title={Unet++: A nested u-net architecture for medical image segmentation},
  author={Zhou, Zongwei and Rahman Siddiquee, Md Mahfuzur and Tajbakhsh, Nima and Liang, Jianming},
  booktitle={International workshop on deep learning in medical image analysis},
  pages={3--11},
  year={2018},
  organization={Springer}
}

@article{wang2004image,
  title={Image quality assessment: from error visibility to structural similarity},
  author={Wang, Zhou and Bovik, Alan C and Sheikh, Hamid R and Simoncelli, Eero P},
  journal={IEEE transactions on image processing},
  volume={13},
  number={4},
  pages={600--612},
  year={2004},
  publisher={IEEE}
}

@article{heusel2017gans,
  title={Gans trained by a two time-scale update rule converge to a local nash equilibrium},
  author={Heusel, Martin and Ramsauer, Hubert and Unterthiner, Thomas and Nessler, Bernhard and Hochreiter, Sepp},
  journal={Advances in neural information processing systems},
  volume={30},
  year={2017}
}

@article{carlsson2009topology,
  title={Topology and data},
  author={Carlsson, Gunnar},
  journal={Bulletin of the American Mathematical Society},
  volume={46},
  number={2},
  pages={255--308},
  year={2009}
}

@article{liu2025persguard,
  title={PersGuard: Preventing Malicious Personalization via Backdoor Attacks on Pre-trained Text-to-Image Diffusion Models},
  author={Liu, Xinwei and Jia, Xiaojun and Xun, Yuan and Zhang, Hua and Cao, Xiaochun},
  journal={arXiv preprint arXiv:2502.16167},
  year={2025}
}

@inproceedings{guo2025copyrightshield,
  title={CopyrightShield: Enhancing Diffusion Model Security Against Copyright Infringement Attacks},
  author={Guo, Zhixiang and Liang, Siyuan and Liu, Aishan and Tao, Dacheng},
  booktitle={Proceedings of the IEEE/CVF International Conference on Computer Vision},
  pages={19417--19426},
  year={2025}
}

@inproceedings{zhang2025patfinger,
  title={Patfinger: Prompt-adapted transferable fingerprinting against unauthorized multimodal dataset usage},
  author={Zhang, Wenyi and Jia, Ju and Jia, Xiaojun and Huang, Yihao and Li, Xinfeng and Wu, Cong and Wang, Lina},
  booktitle={Proceedings of the 48th International ACM SIGIR Conference on Research and Development in Information Retrieval},
  pages={403--413},
  year={2025}
}

@inproceedings{liu2022watermark,
  title={Watermark vaccine: Adversarial attacks to prevent watermark removal},
  author={Liu, Xinwei and Liu, Jian and Bai, Yang and Gu, Jindong and Chen, Tao and Jia, Xiaojun and Cao, Xiaochun},
  booktitle={European conference on computer vision},
  pages={1--17},
  year={2022},
  organization={Springer}
}

@inproceedings{chen2023universal,
  title={Universal watermark vaccine: Universal adversarial perturbations for watermark protection},
  author={Chen, Jianbo and Liu, Xinwei and Liang, Siyuan and Jia, Xiaojun and Xun, Yuan},
  booktitle={Proceedings of the IEEE/CVF Conference on Computer Vision and Pattern Recognition},
  pages={2322--2329},
  year={2023}
}

@inproceedings{chen2022copy,
  title={Copy, right? a testing framework for copyright protection of deep learning models},
  author={Chen, Jialuo and Wang, Jingyi and Peng, Tinglan and Sun, Youcheng and Cheng, Peng and Ji, Shouling and Ma, Xingjun and Li, Bo and Song, Dawn},
  booktitle={2022 IEEE symposium on security and privacy (SP)},
  pages={824--841},
  year={2022},
  organization={IEEE}
}

@inproceedings{zhang2018protecting,
  title={Protecting intellectual property of deep neural networks with watermarking},
  author={Zhang, Jialong and Gu, Zhongshu and Jang, Jiyong and Wu, Hui and Stoecklin, Marc Ph and Huang, Heqing and Molloy, Ian},
  booktitle={Proceedings of the 2018 on Asia conference on computer and communications security},
  pages={159--172},
  year={2018}
}

@inproceedings{tao2022better,
  title={Better trigger inversion optimization in backdoor scanning},
  author={Tao, Guanhong and Shen, Guangyu and Liu, Yingqi and An, Shengwei and Xu, Qiuling and Ma, Shiqing and Li, Pan and Zhang, Xiangyu},
  booktitle={Proceedings of the IEEE/CVF Conference on Computer Vision and Pattern Recognition},
  pages={13368--13378},
  year={2022}
}

@inproceedings{hao2024diff,
  title={Diff-Cleanse: Identifying and Mitigating Backdoor Attacks in Diffusion Models},
  author={Jiang, Hao and Xiao, Jin and Hu, Xiaoguang and Chen, Tianyou and Zhao, Jiajia},
  booktitle={2025 IEEE International Conference on Multimedia and Expo (ICME)},
  pages={1--6},
  year={2025},
  organization={IEEE}
}

@article{lai2025principles,
  title={The principles of diffusion models},
  author={Lai, Chieh-Hsin and Song, Yang and Kim, Dongjun and Mitsufuji, Yuki and Ermon, Stefano},
  journal={arXiv preprint arXiv:2510.21890},
  year={2025}
}

@book{bishop2023deep,
  title={Deep learning: Foundations and concepts},
  author={Bishop, Christopher M and Bishop, Hugh},
  year={2023},
  publisher={Springer Nature}
}

@inproceedings{liu2019variance,
  title={On the variance of the adaptive learning rate and beyond},
  author={Liu, Liyuan and Jiang, Haoming and He, Pengcheng and Chen, Weizhu and Liu, Xiaodong and Gao, Jianfeng and Han, Jiawei},
  booktitle={8th International Conference on Learning Representations, ICLR 2020},
  year={2020}
}

@inproceedings{loshchilov2017decoupled,
  title={Decoupled weight decay regularization},
  author={Loshchilov, Ilya and Hutter, Frank},
  booktitle={International Conference on Learning Representations},
  year={2019}
}

@inproceedings{xie2017aggregated,
  title={Aggregated residual transformations for deep neural networks},
  author={Xie, Saining and Girshick, Ross and Doll{\'a}r, Piotr and Tu, Zhuowen and He, Kaiming},
  booktitle={Proceedings of the IEEE conference on computer vision and pattern recognition},
  pages={1492--1500},
  year={2017}
}

@article{groll2025separation,
  title={Separation of duty in information security},
  author={Groll, Sebastian and Fuchs, Ludwig and Pernul, G{\"u}nther},
  journal={ACM Computing Surveys},
  volume={57},
  number={7},
  pages={1--35},
  year={2025},
  publisher={ACM New York, NY}
}

@misc{barker2007sp,
  title={Sp 800-57. recommendation for key management, part 1: General (revised)},
  author={Barker, Elaine B and Barker, William C and Burr, William E and Polk, W Timothy and Smid, Miles E},
  year={2007},
  publisher={National Institute of Standards \& Technology}
}

@inproceedings{wang2024eviledit,
  title={Eviledit: Backdooring text-to-image diffusion models in one second},
  author={Wang, Hao and Guo, Shangwei and He, Jialing and Chen, Kangjie and Zhang, Shudong and Zhang, Tianwei and Xiang, Tao},
  booktitle={Proceedings of the 32nd ACM International Conference on Multimedia},
  pages={3657--3665},
  year={2024}
}

@inproceedings{lin2014microsoft,
  title={Microsoft coco: Common objects in context},
  author={Lin, Tsung-Yi and Maire, Michael and Belongie, Serge and Hays, James and Perona, Pietro and Ramanan, Deva and Doll{\'a}r, Piotr and Zitnick, C Lawrence},
  booktitle={European conference on computer vision},
  pages={740--755},
  year={2014},
  organization={Springer}
}

@inproceedings{hessel2021clipscore,
  title={Clipscore: A reference-free evaluation metric for image captioning},
  author={Hessel, Jack and Holtzman, Ari and Forbes, Maxwell and Le Bras, Ronan and Choi, Yejin},
  booktitle={Proceedings of the 2021 conference on empirical methods in natural language processing},
  pages={7514--7528},
  year={2021}
}

@inproceedings{wang2024kill,
  title={Kill two birds with one stone: Rethinking data augmentation for deep long-tailed learning},
  author={Wang, Binwu and Wang, Pengkun and Xu, Wei and Wang, Xu and Zhang, Yudong and Wang, Kun and Wang, Yang},
  booktitle={the twelfth international conference on learning representations},
  year={2024}
}

@article{wang2024llm,
  title={Llm-autoda: Large language model-driven automatic data augmentation for long-tailed problems},
  author={Wang, Pengkun and Zhao, Zhe and Wen, HaiBin and Wang, Fanfu and Wang, Binwu and Zhang, Qingfu and Wang, Yang},
  journal={Advances in Neural Information Processing Systems},
  volume={37},
  pages={64915--64941},
  year={2024}
}

@inproceedings{zhao2024two,
  title={Two fists, one heart: Multi-objective optimization based strategy fusion for long-tailed learning},
  author={Zhao, Zhe and Wang, Pengkun and Wen, HaiBin and Xu, Wei and Lai, Song and Zhang, Qingfu and Wang, Yang},
  booktitle={Forty-first International Conference on Machine Learning},
  year={2024}
}

@inproceedings{zhaobalancing,
  title={Balancing Model Efficiency and Performance: Adaptive Pruner for Long-tailed Data},
  author={Zhao, Zhe and Wen, HaiBin and Wang, Pengkun and Wang, Zhenkun and Zhang, Qingfu and Wang, Yang and others},
  booktitle={Forty-second International Conference on Machine Learning},
  year={2025}
}

@article{zhao2026deciphering,
  title={Deciphering the Extremes: A Novel Approach for Pathological Long-tailed Recognition in Scientific Discovery},
  author={Zhao, Zhe and Wen, HaiBin and Liu, Xianfu and Mao, Rui and Wang, Pengkun and Yu, Liheng and Chen, Linjiang and An, Bo and Zhang, Qingfu and Wang, Yang},
  journal={Advances in Neural Information Processing Systems},
  volume={38},
  pages={139640--139662},
  year={2026}
}

@inproceedings{mao2023cross,
  title={Cross-entropy loss functions: Theoretical analysis and applications},
  author={Mao, Anqi and Mohri, Mehryar and Zhong, Yutao},
  booktitle={International conference on Machine learning},
  pages={23803--23828},
  year={2023},
  organization={pmlr}
}
\bibliographystyle{icml2026}

\newpage
\appendix
\onecolumn

\clearpage
\phantomsection
\begin{center}
    {\Large\bfseries Appendix Contents}
\end{center}
\vspace{1em}

\newcommand{\appendixcontentssection}[2]{%
    \noindent\hyperref[#2]{\textbf{\ref*{#2}\quad #1}}\dotfill\hyperref[#2]{\pageref*{#2}}\par
}
\newcommand{\appendixcontentssubsection}[2]{%
    \noindent\hspace*{1.5em}\hyperref[#2]{\ref*{#2}\quad #1}\dotfill\hyperref[#2]{\pageref*{#2}}\par
}

\appendixcontentssection{Proof}{app:proof}
\appendixcontentssection{Additional Experiments Results}{app:additional-experiments-results}
\appendixcontentssubsection{Details of Datasets}{app:details-of-datasets}
\appendixcontentssubsection{Implementation Details}{app:implementation-details}
\appendixcontentssubsection{Inpainting Results on CelebA}{app:celeba-inpaint}
\appendixcontentssubsection{Results on ImageNet}{app:imagenet-results}
\appendixcontentssubsection{Ablation Study on Signature Injection Strength}{app:ablation-injection-strength}
\appendixcontentssubsection{\emph{GoodDiffusion} against Fine-tuning Attack}{app:ablation-fine-tuning}
\appendixcontentssubsection{Computation Cost Analysis}{app:computation-cost}
\appendixcontentssubsection{Potential Adaptive Attacks}{app:adaptive-attacks}
\appendixcontentssubsection{Proactive Protection for Text-to-Image Diffusion Models}{app:proactive-protection-text-to-image}
\appendixcontentssubsection{Additional Visualization Results}{app:additional-visualizations}

\clearpage
\section{Proof}\label{app:proof}

\begin{lemma}[Posterior Distribution of Diffusion Bridge Model~\cite{zhou2023denoising}]\label{lem-posterior}
    For a diffusion bridge model with linear--Gaussian forward kernel, the posterior distribution $q(\bm x_t \mid \bm x_0, \bm x_1)$ is Gaussian and can be expressed as:
    \begin{equation}\label{eq-posterior}
        \begin{split}
            &q(\bm x_t \mid \bm x_0, \bm x_1) =\mathcal{N}(a_t\bm x_1 + b_t\bm x_0, c_t^2 I)\\
            &a_t = \frac{\alpha_t}{\alpha_1}\frac{\mathrm{SNR}_1}{\mathrm{SNR}_t},\quad b_t=\alpha_t\left(1-\frac{\mathrm{SNR}_1}{\mathrm{SNR}_t}\right), \quad c_t^2 = \sigma_t^2\left(1-\frac{\mathrm{SNR}_1}{\mathrm{SNR}_t}\right),\\
        \end{split}
    \end{equation}
    where $\mathrm{SNR}_t=\frac{\alpha_t^2}{\sigma_t^2}$ is the signal-to-noise ratio at time $t$.
    
\end{lemma}

\begin{assumption}\label{assump-linear-corruption}
    For the Image-to-Image generation tasks, given an input-target image pair $(\bm x_0, \bm x_1)$, there exists a linear deterministic corruption operator $\bm{A}$ such that $\bm x_0 = \bm{A}\bm x_1$, where $\bm A\in \mathbb{R}^{d\times d}$ is invertible.
\end{assumption}

\begin{assumption}\label{assump-gaussian-distribution}
    The corrupted input image $\bm x_0$ and the target image $\bm x_1$ follow Gaussian distributions: $\bm x_0 \sim \mathcal{N}(\bm \mu_0, \bm \Sigma_0)$ and $\bm x_1 \sim \mathcal{N}(\bm \mu_1, \bm \Sigma_1)$.
\end{assumption}

\begin{assumption}\label{assump-score-matching}
    A diffusion bridge model $s_{\bm \theta}(\bm{x}_t,t)$ trained on Eq.~\ref{eq-dsm} can perfectly match the score function of the true diffusion process: $s_{\bm \theta}(\bm{x}_t,t) = \mathbb{E}_{\bm x_0 \sim p(\bm x_0 \mid \bm x_t)} \left[ \nabla_{\bm x_t} \log p(\bm x_t \mid \bm x_0) \right] = \nabla_{\bm x_t} \log p(\bm x_t)$ for almost all $\bm x_t\sim p(\bm x_t)$ and $t\in [0,1]$~\cite{lai2025principles}.
\end{assumption}

\textbf{Theorem~\ref{thm:encryption-recovery}} (White-Box Signature Recovery). 
\textit{To bypass the protection of the \emph{GoodDiffusion} with a static signature $\bm k$,
    one can treat a whole-image perturbation $\bm N\in\mathbb{R}^{H\times W}$ as the surrogate attack variable.
    After the optimization, the score function with the surrogate signature $s_{\bm \theta}(\hat{\bm x}_t,t)$ perfectly matches the score function $s_{\bm \theta}(\tilde{\bm x}_t, t)$ with the true signature, thus the recovered surrogate signature $\bm N^*$ approximates the true signature $\bm k$.}

\begin{proof}
    According to the Assumption~\ref{assump-score-matching}, the well-trained model $s_{\bm \theta}$ can perfectly match the score function of the true diffusion process.
    Thus, we have:
    \begin{equation}
        \begin{split}
            &s_{\bm \theta}(\tilde{\bm x}_t, t) = \nabla_{\bm x_t} \log p(\tilde{\bm x}_t),\\
            &s_{\bm \theta}(\hat{\bm x}_t, t) = \nabla_{\bm x_t} \log p(\hat{\bm x}_t).\\
        \end{split}
    \end{equation}


    Then, for the optimal surrogate perturbation $\bm N^*$, to match the score function, we obtain:
    \begin{equation}\label{eq-optimal-surrogate}
        \nabla_{\bm x_t}\log p(\tilde{\bm x}_t) = \nabla_{\bm x_t}\log p(\hat{\bm x}_t).
    \end{equation}

    According to the Lemma~\ref{lem-posterior}, the intermediate ${\bm x}_t$ can be derived given paired images $(\bm x_0, \bm x_1)$:
    \begin{equation}
        \begin{split}
            p(\bm x_t\mid \bm x_0, \bm x_1) &= \mathcal{N}(a_t\bm x_1 + b_t\bm x_0, c_t^2 I).\\
        \end{split}
    \end{equation}

    Under the Assumption~\ref{assump-linear-corruption}, we can further obtain:
    \begin{equation}
        \begin{split}
            p(\bm x_t\mid \bm x_1) &= \mathcal{N}(a_t\bm x_1 + b_t\bm A\bm x_1, c_t^2 I) = \mathcal{N}((a_t I + b_t\bm A)\bm x_1, c_t^2 I).\\
        \end{split}
    \end{equation}

    As we assume that $p(\bm x_1)$ follows a Gaussian distribution $p(\bm x_1)=\mathcal{N}(\bm \mu_1, \bm \Sigma_1)$ (Assumption~\ref{assump-gaussian-distribution}), we can derive the marginal distribution of $\bm x_t$ as~\cite{bishop2023deep}:
    \begin{equation}
        \begin{split}
            p(\bm x_t) &= \int p(\bm x_t\mid \bm x_1) p(\bm x_1) d\bm x_1\\
            &= \mathcal{N}((a_t I + b_t\bm A)\bm \mu_1, (a_t I + b_t\bm A)\Sigma_1(a_t I + b_t\bm A)^\top + c_t^2 I).\\
        \end{split}
    \end{equation}

    Let
    \begin{equation}
        \begin{split}
            &\bm M_1 = a_t I + b_t\bm A,\\
            &\bm M_2 = (a_t I + b_t\bm A)\Sigma_1(a_t I + b_t\bm A)^\top + c_t^2 I.\\
        \end{split}
    \end{equation}
    Thus, we have:
    \begin{equation}
        \begin{split}
            \nabla_{\bm x_t}\log p(\bm x_t) &= -\bm M_2^{-1}(\bm x_t - \bm M_1 \bm \mu_1).\\
        \end{split}
    \end{equation}

    In the case of a naive static signature, we have $\tilde{\bm x}_1 = \bm x_1 + \bm k$ and $\hat{\bm x}_1 = \bm x_1 + \bm N$.
    Therefore, we can derive:
    \begin{equation}\label{eq-score-static}
        \begin{split}
            &\nabla_{\bm x_t}\log p(\tilde{\bm x}_t) = -\bm M_2^{-1}(\bm \tilde{x}_t - \bm M_1(\bm \mu_1 + \bm k)),\\
            &\nabla_{\bm x_t}\log p(\hat{\bm x}_t) = -\bm M_2^{-1}(\bm \hat{x}_t - \bm M_1(\bm \mu_1 + \bm N)).\\
        \end{split}
    \end{equation}

    As $\bm M_2$ is the covariance matrix of $p(\bm x_t)$, it is positive definite and invertible.
    Therefore, according to Eq.~\ref{eq-optimal-surrogate}, we have:
    \begin{equation}\label{eq-surrogate-equality}
        \begin{split}
            &\bm \tilde{x}_t - \bm M_1(\bm \mu_1 + \bm k) = \bm \hat{x}_t - \bm M_1(\bm \mu_1 + \bm N).\\
        \end{split}
    \end{equation}

    According to the reparameterization of $\hat{\bm x}_t$ and $\tilde{\bm x}_t$ based on Lemma~\ref{lem-posterior}, we have:
    \begin{equation}
        \begin{split}
            &\tilde{\bm x}_t = a_t(\bm x_1 + \bm k) + b_t\bm x_0 + \bm \epsilon,\\
            &\hat{\bm x}_t = a_t(\bm x_1 + \bm N) + b_t\bm x_0 + \bm \epsilon',\\
            &\bm \epsilon, \bm \epsilon' \sim \mathcal{N}(0, c_t^2 I).\\
        \end{split}
    \end{equation}

    Thus, Eq.~\ref{eq-surrogate-equality} can be further derived as:
    \begin{equation}
        \begin{split}
            &a_t(\bm x_1 + \bm k) + b_t\bm x_0 + \bm \epsilon - \bm M_1(\bm \mu_1 + \bm k) - (a_t(\bm x_1 + \bm N) + b_t\bm x_0 + \bm \epsilon' - \bm M_1(\bm \mu_1 + \bm N))\\
            &= a_t\bm k - \bm M_1\bm k - a_t\bm N + \bm M_1\bm N + \bm \epsilon - \bm \epsilon'\\
            &= (a_t I - \bm M_1)(\bm k - \bm N) + (\bm \epsilon - \bm \epsilon') = 0.\\
        \end{split}
    \end{equation}
    In other words, we have:
    \begin{equation}
        \begin{split}
            -b_t\bm A(\bm N - \bm k) = \bm \epsilon - \bm \epsilon'.\\
        \end{split}
    \end{equation}
    This formula should hold for arbitrary $\bm \epsilon$ and $\bm \epsilon'$.
    Thus, we take the expectation on both sides:
    \begin{equation}
        \begin{split}
            &\mathbb{E}[-b_t\bm A(\bm N - \bm k)] = \mathbb{E}[\bm \epsilon - \bm \epsilon'] = 0.\\
        \end{split}
    \end{equation}

    Moreover, since $\bm A$ is invertible (Assumption~\ref{assump-linear-corruption}), we can further derive:
    \begin{equation}
        \begin{split}
            &\bm N = \bm k,
        \end{split}
    \end{equation}
    indicating that the optimal surrogate perturbation $\bm N^*$ approximates the true signature $\bm k$.

\end{proof}

\section{Additional Experiments Results}\label{app:additional-experiments-results}

\subsection{Details of Datasets}\label{app:details-of-datasets}
We conduct experiments on two widely-used datasets: CelebA~\cite{liu2015deep} and ImageNet~\cite{deng2009imagenet}.
\textbf{CelebA} contains over 200k celebrity images with rich annotations.
\textbf{ImageNet} is a large-scale dataset with more than 1 million images across a wide variety of categories.
For both datasets, we resize all images to $256\times256$ resolution for training and evaluation.

\subsection{Implementation Details}\label{app:implementation-details}
We implement our \emph{GoodDiffusion} method based on the publicly available codebases.
The proposed \emph{GoodDiffusion} includes two main components: the diffusion bridge model and the learnable signature network.

For the diffusion bridge models, we consider four representative architectures: DDBM-VP, DDBM-VE, I2SB, and DBIM.
We utilize the official implementations for these models.
The DDBM-VP and DDBM-VE~\cite{zhou2023denoising} and DBIM~\cite{zheng2024diffusion} models are trained for 200k iterations with a batch size of 2 paired images.
We set the learning rate to $1e-4$ and use the RAdam optimizer~\cite{liu2019variance}.
The model of I2SB~\cite{liu20232} is trained for 3000 iterations with a batch size of 256 paired samples, using the AdamW optimizer~\cite{loshchilov2017decoupled} with a learning rate of $5e-5$.

As for the learnable signature network, we adopt a UNet++ architecture~\cite{zhou2018unet++} to generate the signatures, which takes the raw image as input and outputs a signature of the same size.
The encoder of the UNet++ is a pretrained ResNeXt backbone~\cite{xie2017aggregated}, while the decoder is trained from scratch.
As introduced in Sec.~\ref{sec:sample-specific}, the learnable signature network is jointly trained with the diffusion bridge model.
We set $\pi_k=0.5$ in Eq.~\ref{eq-total-loss-sample-specific} to treat the authorized and unauthorized training samples equally.
The signature injection strength $\gamma$ is set to 0.9 for all experiments.

For the signature recovery attack in Sec.~\ref{sec:ablation-secure}, we fix the weights of the diffusion bridge model and initialize the surrogate perturbation $\bm N$ with a standard Gaussian noise.
As we assume the adversary has limited computational resources in Sec.~\ref{sec:threat_model}, we select 10k paired samples from the training set, and optimize the surrogate perturbation $\bm N$ for 100 iterations using the AdamW optimizer with a learning rate of $1e-2$ and a batch size of 256.

The training process is conducted on NVIDIA RTX3090 GPUs with 24GB of memory.
We apply the model parallelism technique to distribute the diffusion bridge model and the signature network on different GPUs.

\begin{table*}[!t]
\centering
\caption{\textbf{CelebA Inpainting (Appendix): \textcolor{warn-red}{Protection Effectiveness (PE)} vs \textcolor{auth-blue}{Generation Quality (GQ).}}
We report \textbf{N} (Normal) and \textbf{GD} (GoodDiffusion) for four diffusion bridge models on the Inpainting task.
\textbf{PE} is measured by \textbf{Abuse Rate} (AR). \textbf{Error Rate} (ER) denotes the fraction of authorized inputs incorrectly mapped to the warning image.}
\label{tab:celeba-pe-gq-er-inpaint}

\setlength{\tabcolsep}{4pt}
\renewcommand{\arraystretch}{1.15}

\begin{tabular*}{\textwidth}{@{\extracolsep{\fill}}ll|ccccc@{}}
\toprule[1.5pt]
Model & Setting &
\textcolor{warn-red}{Abuse Rate}$\downarrow$ & \textcolor{auth-blue}{FID}$\downarrow$ & \textcolor{auth-blue}{PSNR}$\uparrow$ & \makecell{\textcolor{auth-blue}{SSIM}\\\textcolor{auth-blue}{(E-02)}}$\uparrow$ & \textcolor{auth-blue}{Error Rate}$\downarrow$ \\
\midrule

\multirow{2}{*}{DDBM-VP}
& Unprotected  & 100 & 23.92 & 24.17 & 86.31 & -- \\
& \textit{GoodDiffusion} & 0   & 34.32 & 22.67 & 85.44 & 0 \\
\midrule

\multirow{2}{*}{DDBM-VE}
& Unprotected  & 100 & 28.46 & 24.10 & 86.72 & -- \\
& \textit{GoodDiffusion} & 0   & 31.20 & 23.33 & 84.90 & 0 \\
\midrule

\multirow{2}{*}{I2SB}
& Unprotected  & 100 & 17.39 & 25.15 & 88.54 & -- \\
& \textit{GoodDiffusion} & 0   & 18.51 & 24.39 & 88.16 & 0.06 \\
\midrule

\multirow{2}{*}{DBIM}
& Unprotected  & 100 & 13.50 & 24.24 & 88.65 & -- \\
& \textit{GoodDiffusion} & 0   & 29.16 & 22.07 & 85.80 & 0 \\

\bottomrule[1.5pt]
\end{tabular*}
\end{table*}

\subsection{Inpainting Results on CelebA}\label{app:celeba-inpaint}

We present the detailed results of \textit{GoodDiffusion} on CelebA Inpainting in Table~\ref{tab:celeba-pe-gq-er-inpaint}.
Similar to the observations on super-resolution and deblurring tasks, \textit{GoodDiffusion} achieves a low \textbf{Abuse Rate} (AR) for unauthorized inputs, while generating high-quality images for authorized inputs with minimal \textbf{Error Rate} (ER).

\subsection{Results on ImageNet}\label{app:imagenet-results}
We present the detailed results of \textcolor{warn-red}{Protection Effectiveness (PE)} and \textcolor{auth-blue}{Generation Quality (GQ)} on ImageNet in Table~\ref{tab:imagenet-pe-gq-er}.

For the \textbf{Protection Effectiveness}, we observe that \emph{GoodDiffusion} consistently achieves an \textbf{Abuse Rate (AR)} of 0\% across all bridge models and tasks.
This shows that our method effectively prevents unauthorized usage.
For the \textbf{Generation Quality}, we find that \emph{GoodDiffusion} maintains competitive performance compared to the normal (unprotected) bridge models.
While there is a slight increase in \textbf{FID} scores for some tasks, the \textbf{PSNR} and \textbf{SSIM} metrics remain relatively stable.
Moreover, the \textbf{Error Rate} is kept very low, indicating that the authorized generations are not affected by the protection mechanism.

\subsection{Ablation Study on Signature Injection Strength}\label{app:ablation-injection-strength}
We conduct an ablation study on the signature injection strength $\gamma$ in Eq.~\ref{eq-sample-specific-injection} to analyze its impact on the protection effectiveness and generation quality.
We vary $\gamma$ from 0.1 to 0.99 and evaluate the performance on the CelebA super-resolution task with the I2SB bridge model in Table~\ref{tab:gamma_results}.

The results show that as $\gamma$ decreases, the quality of the generated images degrades, which aligns with our intuition that a smaller $\gamma$ leads to a stronger signature injection, thus ruining the information in the input images. However, if $\gamma$ is too large (e.g., $\gamma = 0.99$), the model is hard to distinguish the authorized and unauthorized inputs as the signature injection is too weak, thus leading to a high AR/ER and affecting the generation quality. Overall, to achieve a good generation quality, we choose $\gamma = 0.9$ in our experiments.

\begin{table}[t]
\centering
\begin{tabular}{lccccc}
\toprule[1.5pt]
 & \textcolor{warn-red}{Abuse Rate}$\downarrow$ & \textcolor{auth-blue}{FID}$\downarrow$ & \textcolor{auth-blue}{PSNR}$\uparrow$ & \makecell{\textcolor{auth-blue}{SSIM}\\\textcolor{auth-blue}{(E-02)}}$\uparrow$ & \textcolor{auth-blue}{Error Rate}$\downarrow$ \\
\midrule
Unprotected & 100 & 22.05 & 32.48 & 89.93 & - \\
\midrule
$\gamma=0.99$ & 15.88 & 32.93 & 28.38 & 79.48 & 8.44 \\
$\gamma=0.9$ & 0 & \textbf{28.25} & \textbf{30.72} & \textbf{86.18} & 0.25 \\
$\gamma=0.7$ & 0 & 33.65 & 27.18 & 79.93 & 0 \\
$\gamma=0.5$ & 0 & 37.24 & 25.66 & 75.88 & 0 \\
$\gamma=0.3$ & 0 & 42.91 & 23.41 & 70.52 & 0 \\
$\gamma=0.1$ & 0 & 57.38 & 19.61 & 61.83 & 0 \\
\bottomrule[1.5pt]
\end{tabular}
\caption{Performance comparison under different $\gamma$ values.}
\label{tab:gamma_results}
\end{table}

\subsection{\emph{GoodDiffusion} against Fine-tuning Attack}\label{app:ablation-fine-tuning}

To evaluate the robustness of \emph{GoodDiffusion} against fine-tuning attacks, we fine-tune a well-trained \emph{GoodDiffusion} model for 200 steps with a batch size of 256 image pairs.
For comparison, we also train a randomly initialized model with the same training settings.

The results are shown in Table~\ref{tab:fine-tuning-results}.
We observe that the protection of \emph{GoodDiffusion} becomes ineffective after fine-tuning for 50 steps, as the AR reaches 100\%.
However, the model trained from scratch always performs better than the fine-tuned \emph{GoodDiffusion} model in terms of generation quality, which indicates that the fine-tuning process does not fully recover the performance of the original unprotected model.

Overall, we assume that \textbf{the infringers may not have the resources to fine-tune the diffusion model.} Even if they do, the performance of the fine-tuned model is still worse than a model trained from scratch. Thus, \textbf{the infringers do not have to steal the model but train a new one from scratch.}

\begin{table}[t]
\centering
\begin{tabular}{c l c c c c}
\toprule[1.5pt]
\textbf{Step} & \textbf{Model} & \textcolor{warn-red}{Abuse Rate}$\downarrow$ & \textcolor{auth-blue}{FID}$\downarrow$ & \textcolor{auth-blue}{PSNR}$\uparrow$ & \makecell{\textcolor{auth-blue}{SSIM}\\\textcolor{auth-blue}{(E-02)}}$\uparrow$ \\
\midrule
0   & GoodDiffusion         & 0   & 378.82 & 5.86  & 4.83  \\
\midrule
1   & Trained from scratch  & 100 & 97.83  & 23.57 & 43.19 \\
1   & GoodDiffusion         & 0   & 375.49 & 5.91  & 5.21  \\
\midrule
50  & Trained from scratch  & 100 & 22.44  & 31.33 & 88.02 \\
50  & GoodDiffusion         & 100 & 75.17  & 26.27 & 79.92 \\
\midrule
100 & Trained from scratch  & 100 & 24.58  & 31.93 & 88.95 \\
100 & GoodDiffusion         & 100 & 45.15  & 30.33 & 85.84 \\
\midrule
150 & Trained from scratch  & 100 & 25.09  & 32.12 & 89.27 \\
150 & GoodDiffusion         & 100 & 27.37  & 31.30 & 87.40 \\
\midrule
200 & Trained from scratch  & 100 & 25.19  & 32.21 & 89.44 \\
200 & GoodDiffusion         & 100 & 22.50  & 31.72 & 86.18 \\
\bottomrule[1.5pt]
\end{tabular}
\caption{Performance comparison between training from scratch and fine-tuning with \emph{GoodDiffusion} on the CelebA Super-Resolution task with the I2SB bridge model. We report the \textbf{Abuse Rate (AR)}, \textbf{FID}, \textbf{PSNR}, and \textbf{SSIM} at different training steps.}
\label{tab:fine-tuning-results}
\end{table}

\subsection{Computation Cost Analysis}\label{app:computation-cost}
We analyze the computation cost of \emph{GoodDiffusion} in terms of training time and inference time shown in Table~\ref{tab:efficiency}.
The results are obtained by running the model on an RTX 3090 GPU with a batch size of 16 and a resolution of 256x256 for 1 inference step.

The results show that the LSN does not bring significant computational overhead. In addition, as image generation requires a number of steps for the diffusion model, but only one inference for the LSN, the additional computation and time cost of LSN are negligible.

\begin{table}[t]
\centering
\begin{tabular}{l c c}
\toprule[1.5pt]
\textbf{Model} & \textbf{GPU Memory (MB)} & \textbf{Inference Time (ms)} \\
\midrule
DDBM-VP & 7409.07  & 1151.35 \\
DDBM-VE & 7407.07  & 1145.27 \\
I2SB    & 10981.84 & 933.07  \\
DDIM    & 7407.07  & 1189.45 \\
\midrule
LSN     & 4401.11  & 206.98  \\
\bottomrule[1.5pt]
\end{tabular}
\caption{\textbf{Efficiency Analysis.} We report the GPU memory usage and inference time of the diffusion bridge models and the learnable signature network (LSN) on the CelebA Super-Resolution task. The results are measured on an NVIDIA RTX3090 GPU with a batch size of 16 and an image resolution of $256\times256$ for 1 inference step.}
\label{tab:efficiency}
\end{table}

\subsection{Potential Adaptive Attacks}\label{app:adaptive-attacks}
We discuss potential adaptive attacks that infringers may attempt to bypass the protection of \emph{GoodDiffusion}.
Specifically, we randomly initialize a surrogate LSN $g_{\phi'}$ and freeze the parameters of a well-trained I2SB \emph{GoodDiffusion} model $s_\theta$.
A straightforward attack is to optimize the surrogate LSN to minimize the following recovery loss, modified from Eq.~\ref{eq-recovery-loss}:

\begin{equation}
    \begin{split}
        \mathcal{L}(\phi') &= \mathbb{E}_{(\bm x_1, \bm x_0)\sim D_a, t\sim\mathcal{U}(0,1)} \left[ \lambda(t) \left\| s_{\bm \theta}(\hat{\bm{x}}_t,t) - s^*(\hat{\bm{x}}_t,t) \right\|_2^2 \right],\\
        \hat{\bm x}_1 &= \gamma\bm x_1 + (1-\gamma)g_{\phi'}(\bm x_1),\\
        \hat{\bm x}_t &\sim q(\bm x_t\mid \hat{\bm x}_1, \bm x_0).
    \end{split}
\end{equation}

We train the surrogate LSN $g_{\phi'}$ for 200 steps with a batch size of 256, and compare the performance of the model with the original LSN (O-LSN) and the surrogate LSN (S-LSN).
As shown in Table~\ref{tab:adaptive-attack-results}, the surrogate LSN fails to substitute the original LSN to bypass the protection, as the model can only generate the warning images.

The reason for the failure of the surrogate LSN may be that the GoodDiffusion model is trained to exhibit a threshold behavior: unless the transformed input is close enough to the authorized manifold, the model will generate warning images.
\textbf{Such behavior is hard to learn without the supervision of the signature service.}

\begin{table}[t]
\centering
\begin{tabular}{l c c c c}
\toprule[1.5pt]
\textbf{Model} & \textcolor{auth-blue}{FID}$\downarrow$ & \textcolor{auth-blue}{PSNR}$\uparrow$ & \makecell{\textcolor{auth-blue}{SSIM}\\\textcolor{auth-blue}{(E-02)}}$\uparrow$ & \textcolor{auth-blue}{Error Rate}$\downarrow$ \\
\midrule
O-LSN & 28.25 & 30.72 & 86.18 & 0.25 \\
S-LSN & 376.91 & 6.04 & 5.42 & 100 \\
\bottomrule[1.5pt]
\end{tabular}
\caption{\textbf{Performance of Adaptive Surrogate Attack.} We compare the performance of the original LSN (O-LSN) and the surrogate LSN (S-LSN) on the CelebA Super-Resolution task with the I2SB bridge model. The results show that the surrogate LSN fails to bypass the protection.}
\label{tab:adaptive-attack-results}
\end{table}

\subsection{Proactive Protection for Text-to-Image Diffusion Models}\label{app:proactive-protection-text-to-image}

In this paper, the proposed \emph{GoodDiffusion} method is designed for diffusion bridge models that conduct image-to-image (I2I) translation tasks, such as super-resolution, inpainting, and deblurring.
However, the core idea of \emph{GoodDiffusion} can be extended to proactively protect text-to-image (T2I) diffusion models as well.

Here, we show some preliminary results of applying \emph{GoodDiffusion} to T2I diffusion models.
The core idea is to follow the backdoor attack methods in T2I models~\cite{wang2024eviledit}: a special text prompt is used as the signature, and the model solely generates high-quality images when the signature is present in the text prompt. We compute the FID between the generated images and the randomly selected samples from MS-COCO~\cite{lin2014microsoft}, and the CLIP Score~\cite{hessel2021clipscore} between the generated images and the text prompts.

The results in Table~\ref{tab:t2i-results} show that the model can generate high-quality images for authorized prompts while generating low-quality images for unauthorized prompts, demonstrating the potential of \emph{GoodDiffusion} in protecting T2I diffusion models.

\begin{table}[t]
\centering
\begin{tabular}{l l c c c c}
\toprule[1.5pt]
\textbf{Model} & \textbf{Settings} & \textcolor{auth-blue}{FID}$\downarrow$ & \textcolor{auth-blue}{CLIP Score}$\uparrow$ & \textcolor{auth-blue}{Error Rate}$\downarrow$ & \textcolor{warn-red}{Abuse Rate}$\downarrow$ \\
\midrule
SD V1.4 & Unprotected & 50.99 & 31.22 & - & - \\
SD V1.4	& Authorized & 50.62 & 30.98 & 0.6 & - \\
SD V1.4 & Unauthorized & 284.26 & 17.21 & - & 3.3 \\
\midrule
SD V1.5 & Unprotected & 50.00 & 31.24 & - & - \\
SD V1.5	& Authorized & 50.70 & 31.07 & 0.7 & - \\
SD V1.5	& Unauthorized & 282.63 & 17.83 & - & 4.5 \\
\bottomrule[1.5pt]
\end{tabular}
\caption{\textbf{Preliminary Results of \emph{GoodDiffusion} on Text-to-Image Diffusion Models.} We apply the proactive protection to Stable Diffusion (SD) V1.4 and V1.5 models, where a special text prompt is used as the signature. The results show that the model can generate high-quality images for authorized prompts while generating low-quality images for unauthorized prompts, demonstrating the potential of \emph{GoodDiffusion} in protecting T2I diffusion models.}
\label{tab:t2i-results}
\end{table}

\subsection{Additional Visualization Results}\label{app:additional-visualizations}

We provide additional visualization results of \emph{GoodDiffusion} on CelebA datasets in Fig.~\ref{fig:additional-visualizations-celeba-sr}, Fig.~\ref{fig:additional-visualizations-celeba-inpaint}, Fig.~\ref{fig:additional-visualizations-celeba-deblur}, and on ImageNet datasets in Fig.~\ref{fig:additional-visualizations-imagenet-sr}, Fig.~\ref{fig:additional-visualizations-imagenet-inpaint}, Fig.~\ref{fig:additional-visualizations-imagenet-deblur}.
The results demonstrate that \emph{GoodDiffusion} enables high-quality image generation exclusively when sample-specific signatures are provided, while effectively producing warning images for unauthorized inputs.

We also provide more visualization results for the security analysis in Sec.~\ref{sec:ablation-secure}. In particular, we implement the static signature and sample-specific signature settings on the CelebA dataset with the I2SB bridge model.
The generation results of the two settings are shown in Fig.~\ref{fig:visual-security-analysis-static} and Fig.~\ref{fig:visual-security-analysis-dynamic}, respectively.
The visualizations confirm that an adversary can successfully recover a surrogate signature to bypass the protection in the static signature setting.
However, in the sample-specific signature setting, the adversary fails to recover a reliable surrogate signature, and the generation results for unauthorized inputs are of low quality.
These findings visually show that the sample-specific signature setting provides stronger protection compared to the static signature setting.

\begin{table*}[t]
\centering
\caption{\textbf{ImageNet: \textcolor{warn-red}{Protection Effectiveness (PE)} vs \textcolor{auth-blue}{Generation Quality (GQ).}}
For each task and bridge model, we report two adjacent rows: unprotected model (Unprotected) and \emph{GoodDiffusion} protected model.
The protection effectiveness is measured by \textbf{Abuse Rate} (AR), which indicates the fraction of unauthorized inputs incorrectly mapped to high-quality target images.
As the unprotected model does not perform any protection, the AR is always 100\%.
The generation quality is evaluated by \textbf{FID}, \textbf{PSNR}, and \textbf{SSIM}. We also report \textbf{Error Rate} (ER), which measures the fraction of authorized generation incorrectly blocked to warning images.
}
\label{tab:imagenet-pe-gq-er}
\setlength{\tabcolsep}{4pt}
\renewcommand{\arraystretch}{1.15}
\begin{tabular}{lll|c|cccc}
\toprule[1.5pt]
\multirow{2}{*}{Task} & \multirow{2}{*}{Bridge Model} & \multirow{2}{*}{Setting}
& \multicolumn{1}{c|}{\textcolor{warn-red}{\textbf{PE}}}
& \multicolumn{4}{c}{\textcolor{auth-blue}{\textbf{GQ}}} \\
\cmidrule(lr){4-4}\cmidrule(lr){5-8}
& & &
Abuse Rate$\downarrow$
& FID$\downarrow$ & PSNR$\uparrow$ & \makecell{SSIM\\(E-02)}$\uparrow$ & Error Rate$\downarrow$ \\
\midrule

\multirow{8}{*}{SR}
& \multirow{2}{*}{DDBM-VP} & Unprotected        & 100 & 29.55 & 25.77 & 73.62 & -- \\
&                          & \emph{GoodDiffusion} &  0   & 36.04 & 25.67 & 70.61 & 0   \\
\cmidrule(lr){2-8}
& \multirow{2}{*}{DDBM-VE} & Unprotected        & 100 & 26.73 & 27.11 & 77.94 & -- \\
&                          & \emph{GoodDiffusion} &  0   & 34.50 & 24.05 & 75.34 & 0   \\
\cmidrule(lr){2-8}
& \multirow{2}{*}{I2SB} & Unprotected        & 100 & 28.22 & 27.85 & 78.89 & -- \\
&                          & \emph{GoodDiffusion} &   0  & 31.62 & 26.25 & 74.80 &  0.31  \\
\cmidrule(lr){2-8}
& \multirow{2}{*}{DBIM} & Unprotected        & 100 & 33.56 & 26.15 & 76.71 & -- \\
&                          & \emph{GoodDiffusion} &   0  & 43.41 & 25.40 & 73.65 &  0.31  \\
\midrule

\multirow{8}{*}{Inpaint}
& \multirow{2}{*}{DDBM-VP} & Unprotected        & 100 & 22.28 & 24.60 & 81.97 & -- \\
&                          & \emph{GoodDiffusion} &  0   & 34.81 & 22.38 & 80.64 &  0  \\
\cmidrule(lr){2-8}
& \multirow{2}{*}{DDBM-VE} & Unprotected        & 100 & 27.53 & 23.61 & 81.33 & -- \\
&                          & \emph{GoodDiffusion} &  0   & 37.03 & 21.82 & 77.88 &  0  \\
\cmidrule(lr){2-8}
& \multirow{2}{*}{I2SB} & Unprotected        & 100 & 14.13 & 25.21 & 87.17 & -- \\
&                          & \emph{GoodDiffusion} &  0   & 15.23 & 24.44 & 85.87 &  0  \\
\cmidrule(lr){2-8}
& \multirow{2}{*}{DBIM} & Unprotected        & 100 & 22.91 & 22.18 & 83.86 & -- \\
&                          & \emph{GoodDiffusion} &  0   & 47.45 & 20.32 & 81.12 &  0  \\
\midrule

\multirow{8}{*}{Deblur}
& \multirow{2}{*}{DDBM-VP} & Unprotected        & 100 & 4.82 & 35.98 & 94.89 & -- \\
&                          & \emph{GoodDiffusion} &   0  & 5.70 & 34.10 & 93.23 &  0  \\
\cmidrule(lr){2-8}
& \multirow{2}{*}{DDBM-VE} & Unprotected        & 100 & 8.53 & 34.61 & 94.11 & -- \\
&                          & \emph{GoodDiffusion} &   0  & 17.67 & 30.45 & 86.10 &  0  \\
\cmidrule(lr){2-8}
& \multirow{2}{*}{I2SB} & Unprotected        & 100 & 29.71 & 28.35 & 79.74 & -- \\
&                          & \emph{GoodDiffusion} &  0   & 42.47 & 26.02 & 73.42 &  0  \\
\cmidrule(lr){2-8}
& \multirow{2}{*}{DBIM} & Unprotected        & 100 & 3.31 & 37.69 & 96.74 & -- \\
&                          & \emph{GoodDiffusion} &   0  & 6.45 & 34.79 & 94.89 & 0   \\
\bottomrule[1.5pt]
\end{tabular}
\end{table*}

\begin{figure}
    \centering
    \includegraphics[width=0.7\linewidth]{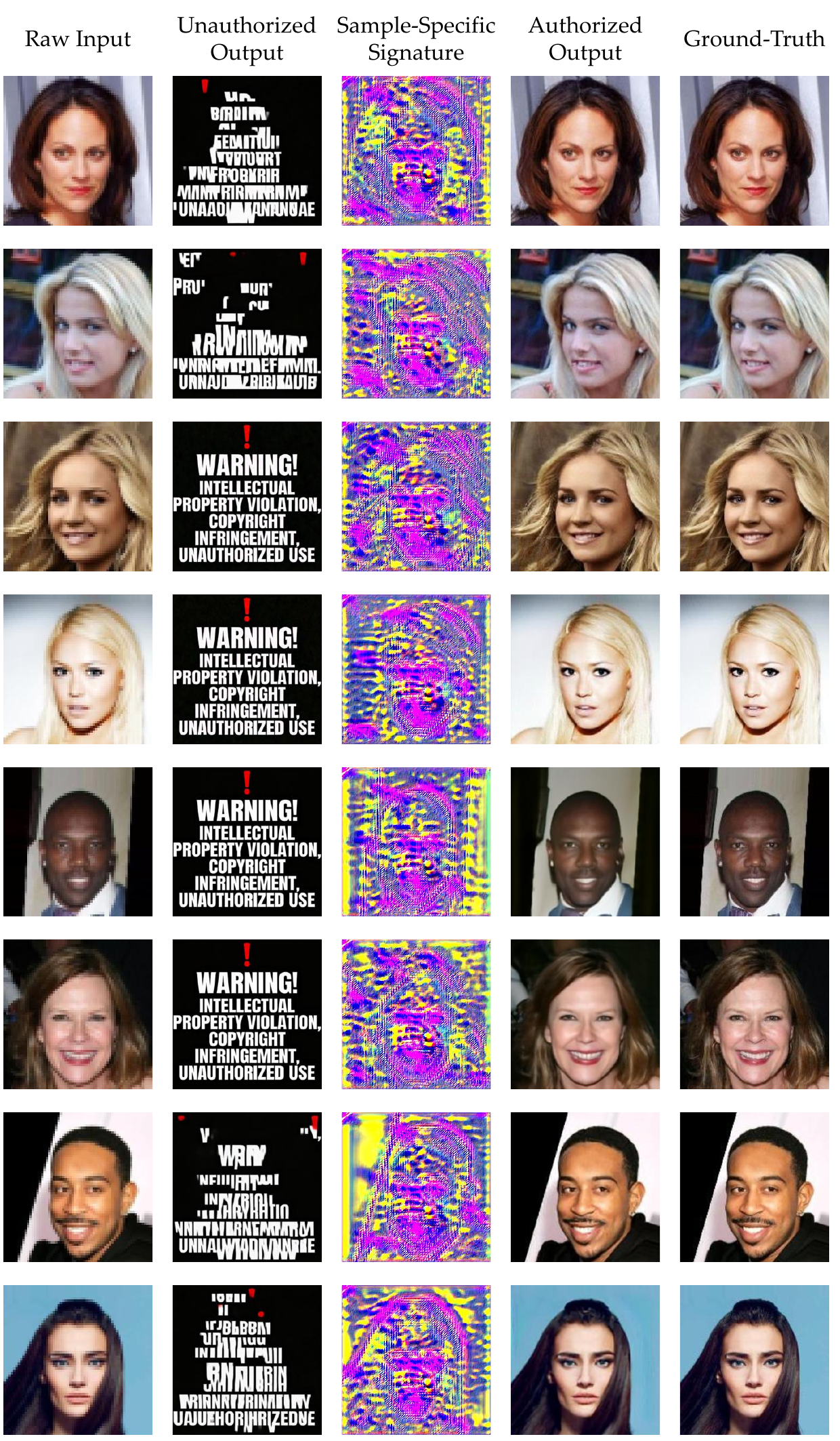}
    \caption{\textbf{Additional Visualization Results on CelebA Super-Resolution.} We present more visualization results of \emph{GoodDiffusion} on the CelebA Super-Resolution task with different diffusion bridge models. \textbf{Top 2 rows}: DDBM-VP. \textbf{3rd and 4th rows}: DDBM-VE. \textbf{5th and 6th rows}: I2SB. \textbf{Bottom 2 rows}: DBIM.}
    \label{fig:additional-visualizations-celeba-sr}
\end{figure}

\begin{figure}
    \centering
    \includegraphics[width=0.7\linewidth]{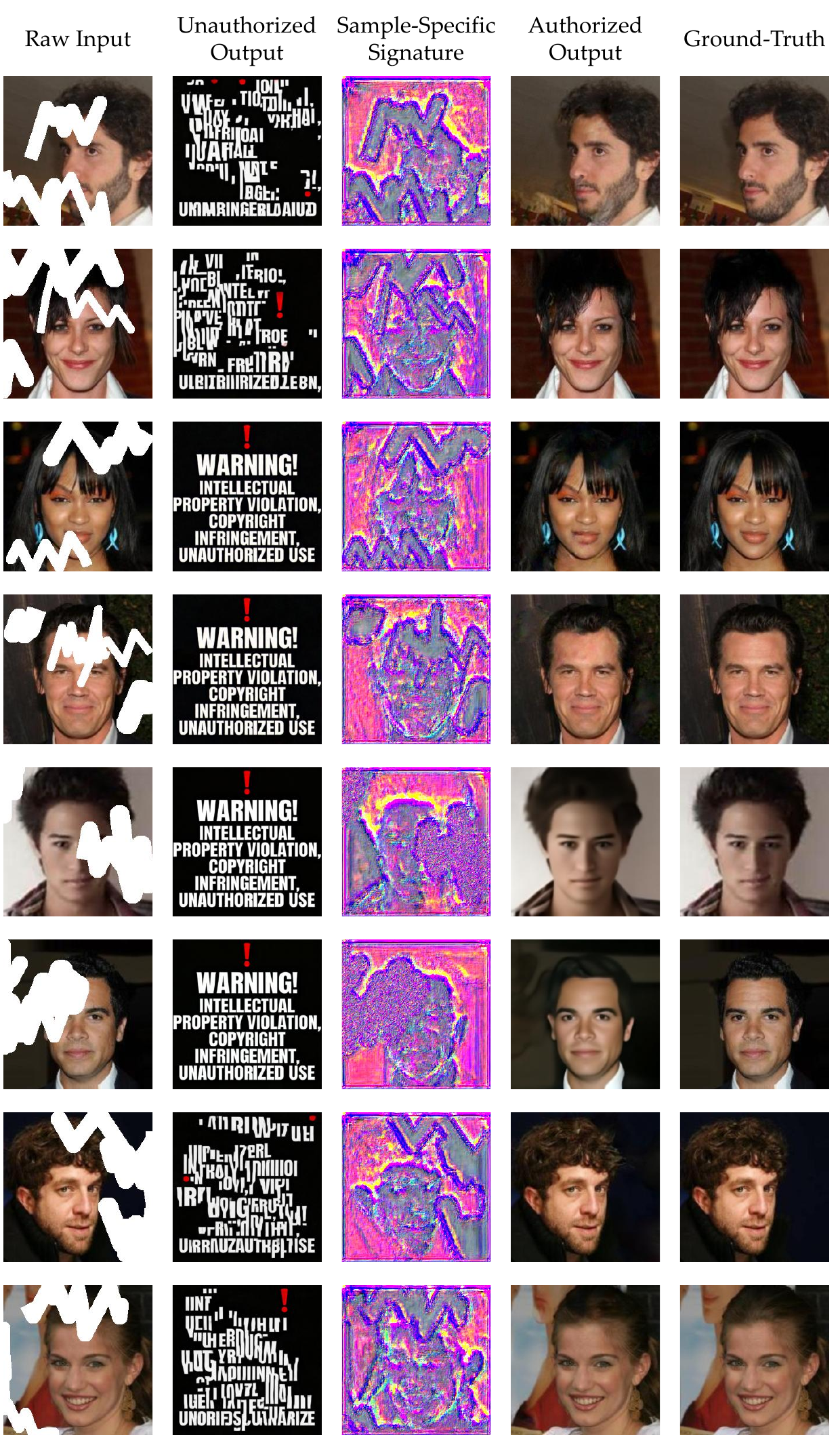}
    \caption{\textbf{Additional Visualization Results on CelebA Inpainting.} We present more visualization results of \emph{GoodDiffusion} on the CelebA Inpainting task with different diffusion bridge models. \textbf{Top 2 rows}: DDBM-VP. \textbf{3rd and 4th rows}: DDBM-VE. \textbf{5th and 6th rows}: I2SB. \textbf{Bottom 2 rows}: DBIM.}
    \label{fig:additional-visualizations-celeba-inpaint}
\end{figure}

\begin{figure}
    \centering
    \includegraphics[width=0.7\linewidth]{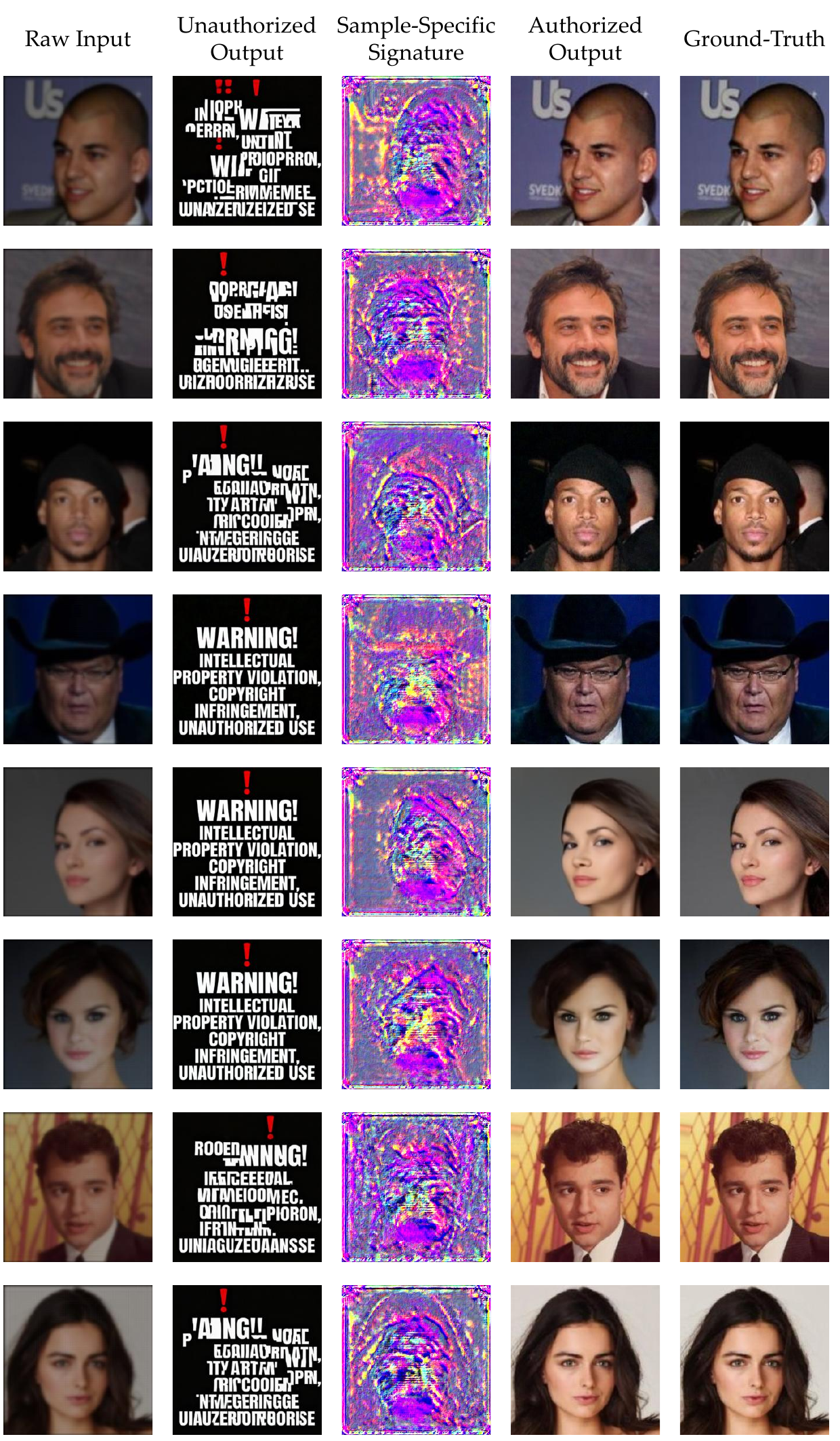}
    \caption{\textbf{Additional Visualization Results on CelebA Deblurring.} We present more visualization results of \emph{GoodDiffusion} on the CelebA Deblurring task with different diffusion bridge models. \textbf{Top 2 rows}: DDBM-VP. \textbf{3rd and 4th rows}: DDBM-VE. \textbf{5th and 6th rows}: I2SB. \textbf{Bottom 2 rows}: DBIM.}
    \label{fig:additional-visualizations-celeba-deblur}
\end{figure}

\begin{figure}
    \centering
    \includegraphics[width=0.7\linewidth]{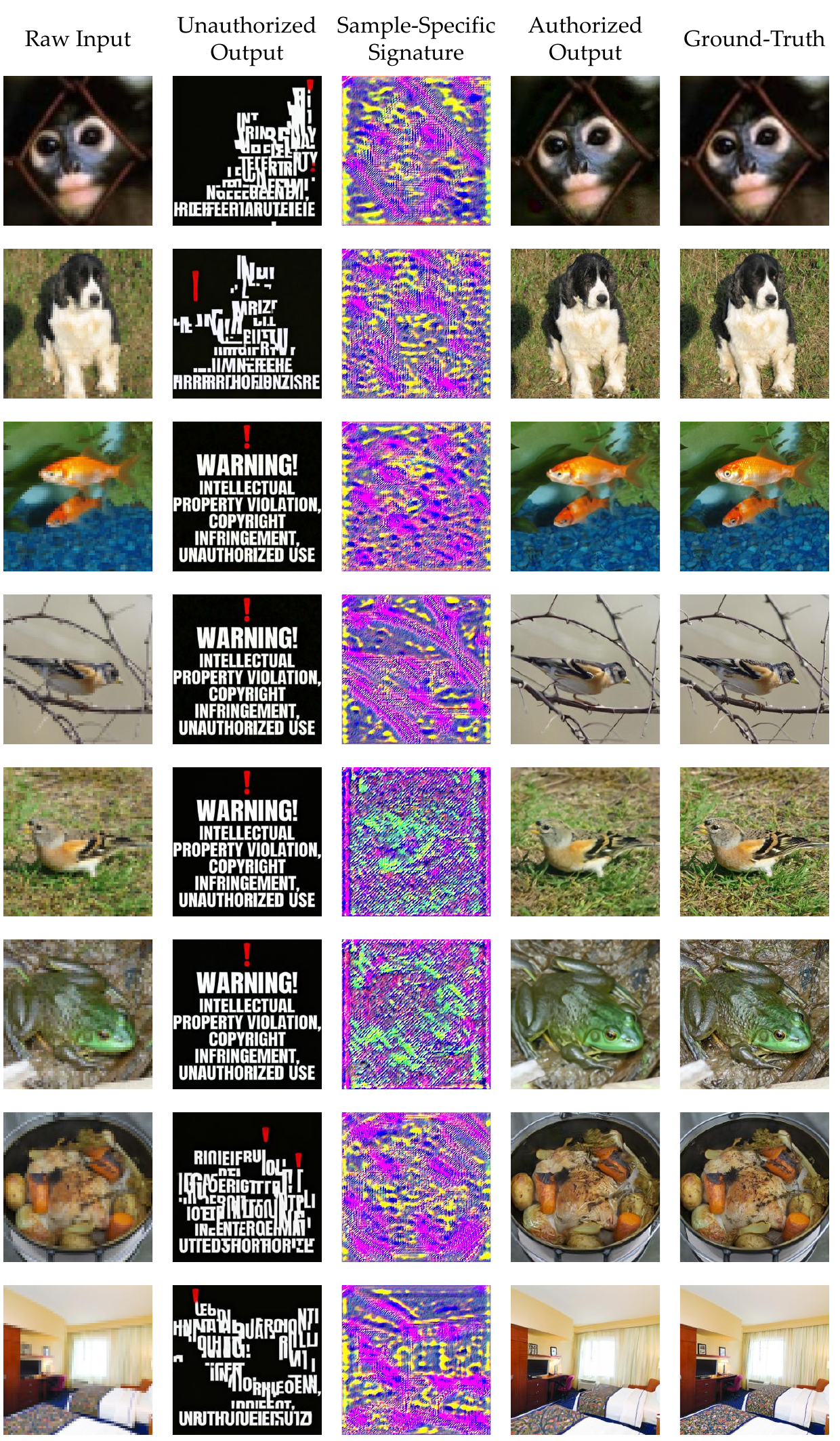}
    \caption{\textbf{Additional Visualization Results on ImageNet Super-Resolution.} We present more visualization results of \emph{GoodDiffusion} on the ImageNet Super-Resolution task with different diffusion bridge models. \textbf{Top 2 rows}: DDBM-VP. \textbf{3rd and 4th rows}: DDBM-VE. \textbf{5th and 6th rows}: I2SB. \textbf{Bottom 2 rows}: DBIM.}
    \label{fig:additional-visualizations-imagenet-sr}
\end{figure}

\begin{figure}
    \centering
    \includegraphics[width=0.7\linewidth]{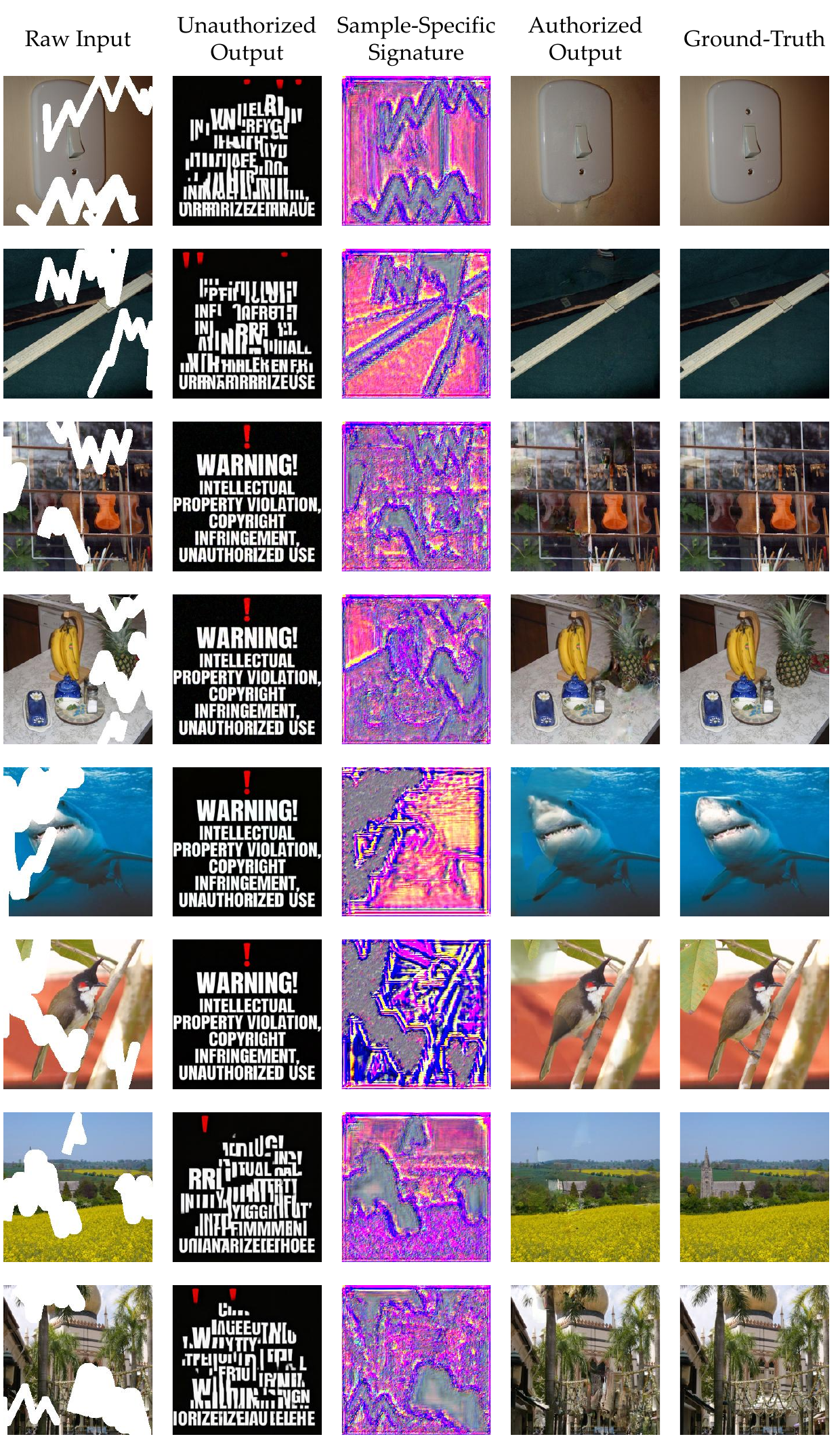}
    \caption{\textbf{Additional Visualization Results on ImageNet Inpainting.} We present more visualization results of \emph{GoodDiffusion} on the ImageNet Inpainting task with different diffusion bridge models. \textbf{Top 2 rows}: DDBM-VP. \textbf{3rd and 4th rows}: DDBM-VE. \textbf{5th and 6th rows}: I2SB. \textbf{Bottom 2 rows}: DBIM.}
    \label{fig:additional-visualizations-imagenet-inpaint}
\end{figure}

\begin{figure}
    \centering
    \includegraphics[width=0.7\linewidth]{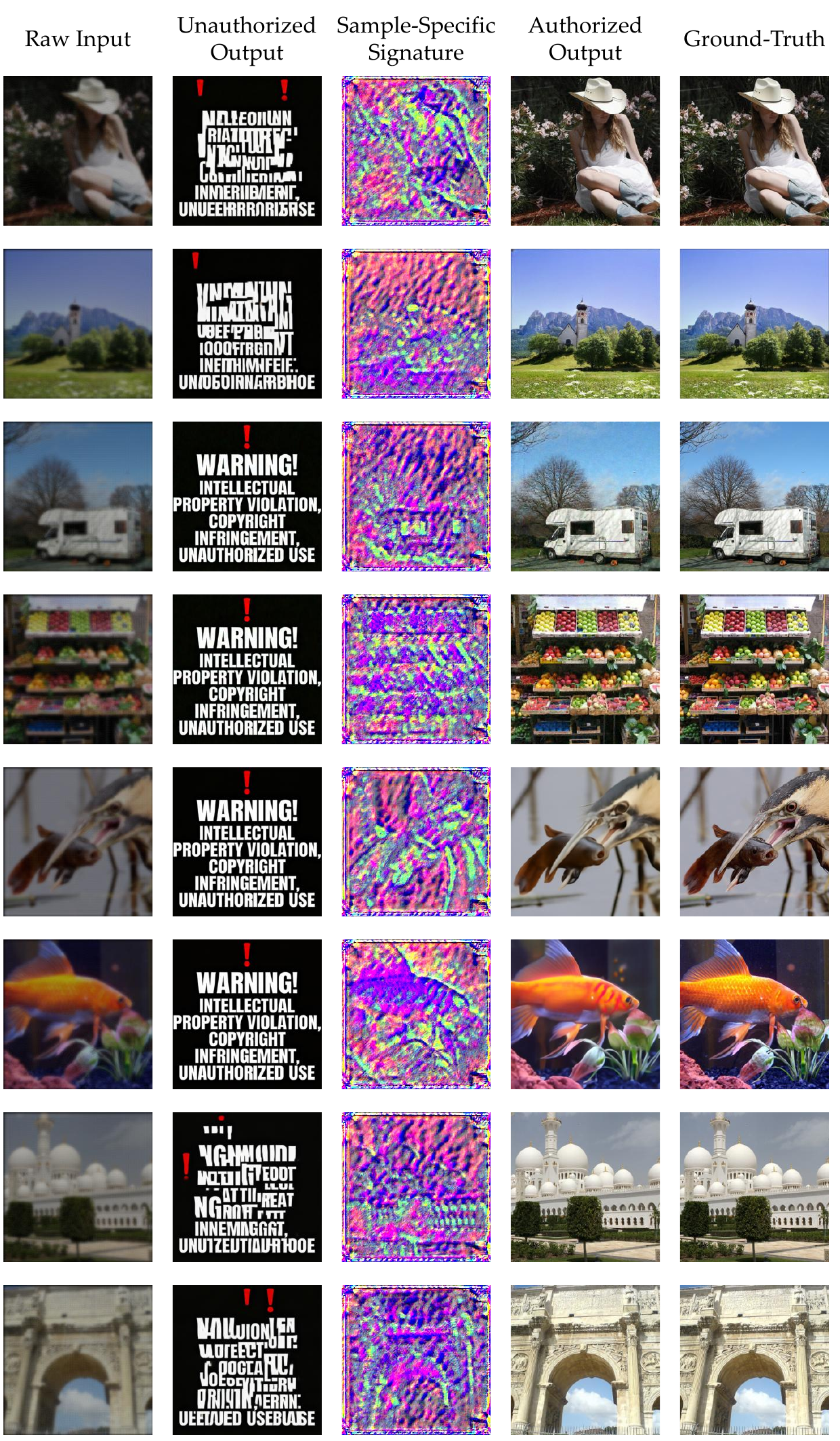}
    \caption{\textbf{Additional Visualization Results on ImageNet Deblurring.} We present more visualization results of \emph{GoodDiffusion} on the ImageNet Deblurring task with different diffusion bridge models. \textbf{Top 2 rows}: DDBM-VP. \textbf{3rd and 4th rows}: DDBM-VE. \textbf{5th and 6th rows}: I2SB. \textbf{Bottom 2 rows}: DBIM.}
    \label{fig:additional-visualizations-imagenet-deblur}
\end{figure}

\begin{figure}
    \centering
    \includegraphics[width=0.8\linewidth]{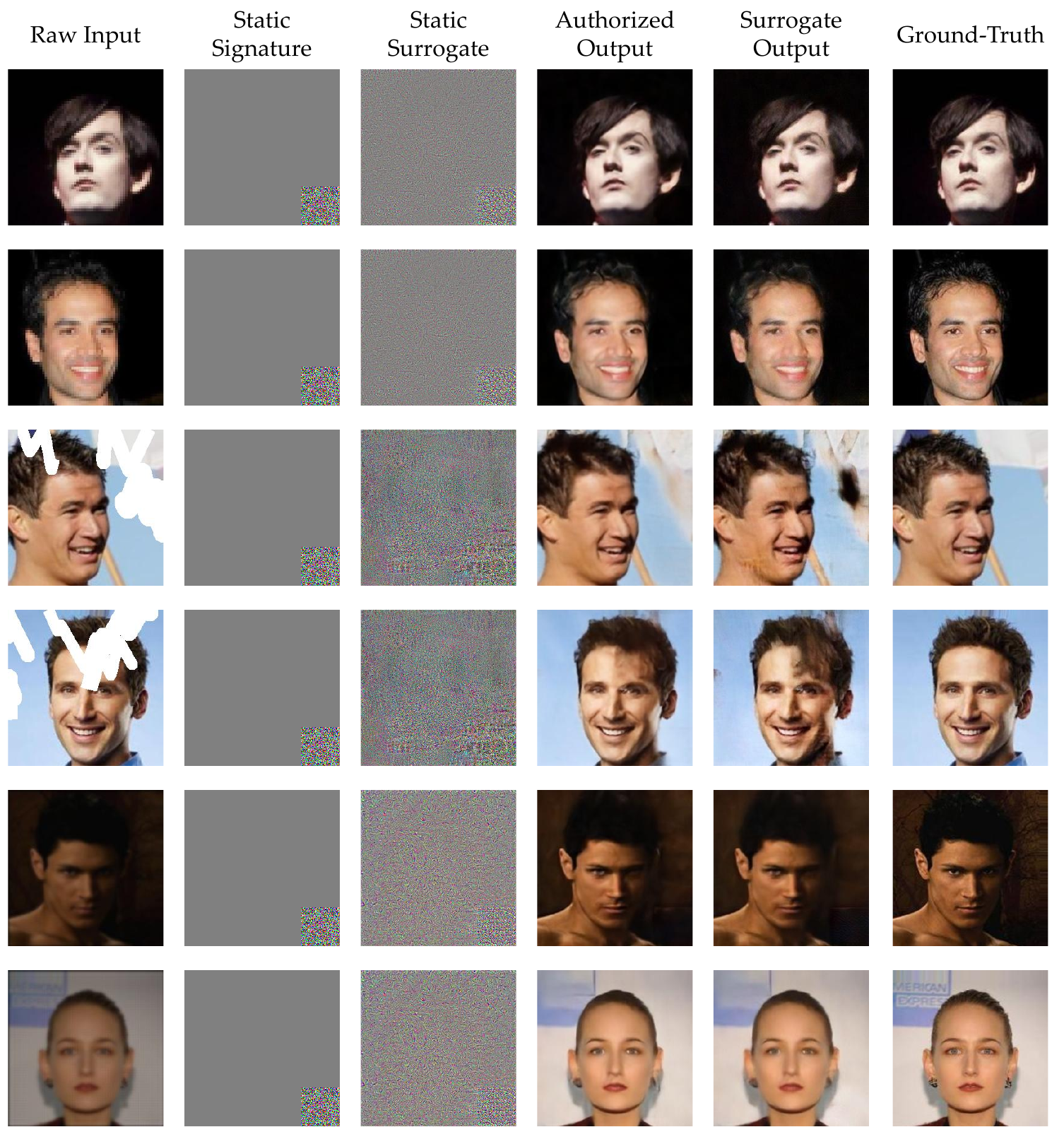}
    \caption{\textbf{Visualization Results for Security Analysis: Static Signature Setting.} We implement the static signature setting on the CelebA dataset with the I2SB bridge model.
    \textbf{Top 2 rows}: Super-Resolution task. \textbf{Middle 2 rows}: Inpainting task. \textbf{Bottom 2 rows}: Deblurring task.
    The results show that the adversary can recover a static surrogate signature that is similar to the true static signature.
    The surrogate signature enables high-quality image generation for unauthorized inputs, indicating that the static signature setting is vulnerable.}
    \label{fig:visual-security-analysis-static}
\end{figure}

\begin{figure}
    \centering
    \includegraphics[width=0.8\linewidth]{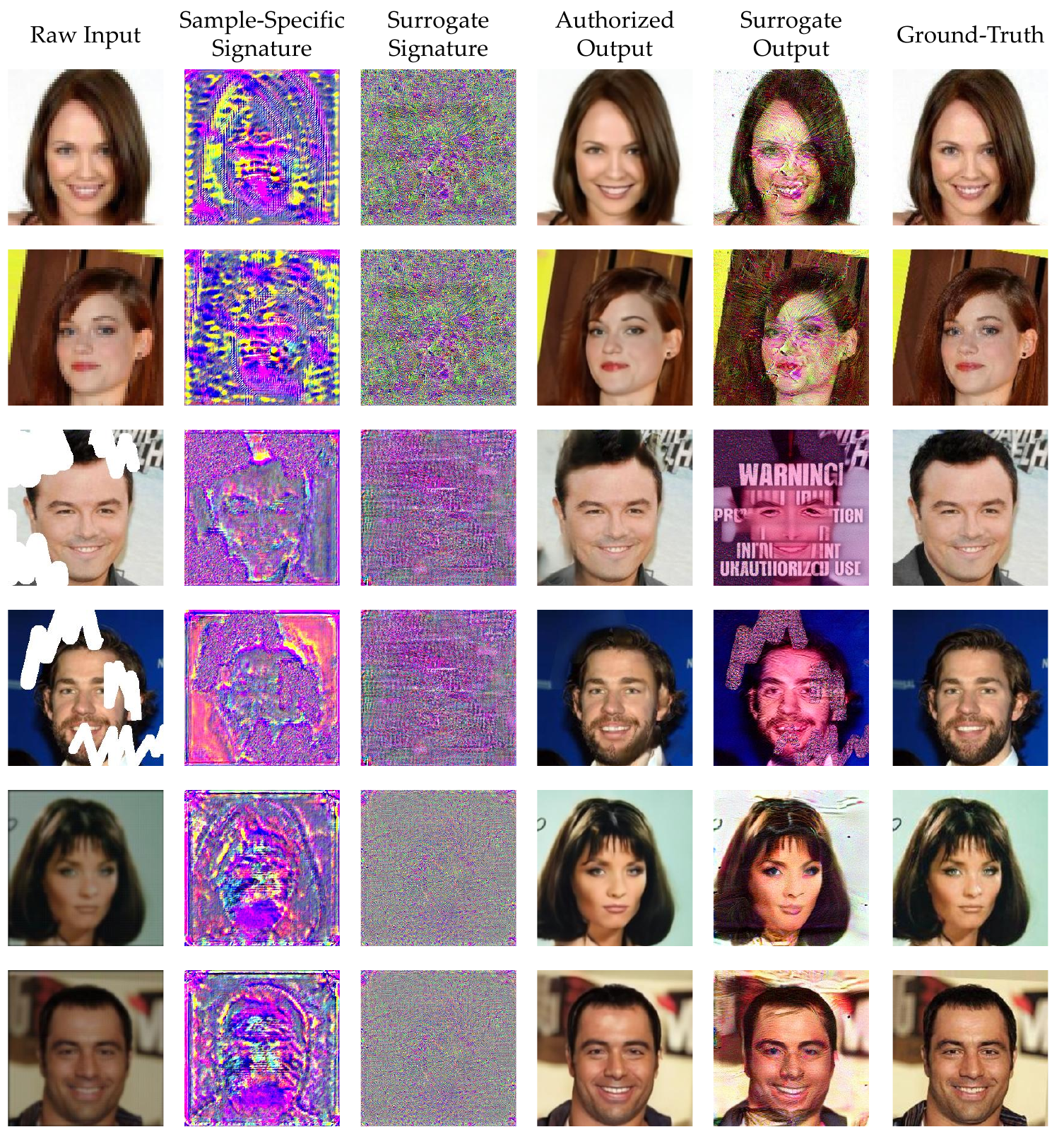}
    \caption{\textbf{Visualization Results for Security Analysis: Sample-Specific Signature Setting.} We implement the sample-specific signature setting on the CelebA dataset with the I2SB bridge model.
    \textbf{Top 2 rows}: Super-Resolution task. \textbf{Middle 2 rows}: Inpainting task. \textbf{Bottom 2 rows}: Deblurring task.
    The results show that the adversary fails to recover a universal surrogate signature.
    The generation results for unauthorized inputs remain low-quality, indicating the sample-specific signature setting is more secure.}
    \label{fig:visual-security-analysis-dynamic}
\end{figure}

\end{document}